\documentclass[12pt,reqno]{amsart}
\usepackage{graphicx}
\usepackage{amssymb,amsmath}
\usepackage{amsthm}
\usepackage{color}
\usepackage{float}
\usepackage[pdf]{pstricks}
\usepackage{hyperref}
\usepackage{empheq}
\numberwithin{equation}{section}

\usepackage{amssymb}
\usepackage{epsfig}
\usepackage{extarrows}
\usepackage{amsfonts}
\usepackage{graphicx}
\usepackage{tikz}
\usepackage{color}
\hypersetup{
	colorlinks=true,
	linkcolor=blue,
	filecolor=magenta,      
	urlcolor=blue,
}

\newcommand{\R}{\mathbb{R}}

\DeclareMathOperator{\sn}{sn}
\DeclareMathOperator{\cn}{cn}
\DeclareMathOperator{\dn}{dn}

\newtheorem{remark}{Remark}
\newtheorem{proposition}{Proposition}
\newtheorem{theorem}{Theorem}
\newtheorem{corollary}{Corollary}
\newtheorem{lemma}{Lemma}

\begin{document}
	
\title[Dark breathers in the defocusing mKdV equation]{\bf Dark breathers on a snoidal wave background \\ in the defocusing mKdV equation}
	
\author{Ana Mucalica}
\address[A. Mucalica]{Department of Mathematics and Statistics, McMaster University,	Hamilton, Ontario, Canada, L8S 4K1}
\email{mucalica@mcmaster.ca} 
	
\author{Dmitry E. Pelinovsky}
\address[D. Pelinovsky]{Department of Mathematics and Statistics, McMaster University, Hamilton, Ontario, Canada, L8S 4K1}
\email{dmpeli@math.mcmaster.ca}

\begin{abstract}
We present a new exact solution to the defocusing modified Korteweg-de Vries equation to describe the interaction of a dark soliton and a traveling periodic wave. The solution (which we refer to as to the dark breather) is obtained by using the Darboux transformation with the eigenfunctions of the Lax system expressed in terms of the Jacobi theta functions. Properties of elliptic functions including the quarter-period translations in the complex plane are applied to transform the solution to the simplest form. We explore the characteristic properties of these dark breathers and show that they propagate faster than the periodic wave (in the same direction) and attain maximal localization at a specific parameter value which is explicitly computed.
\end{abstract}   
     

\maketitle


\section{Introduction}

The main model for this work is the defocusing modified Korteweg--de Vries (mKdV) equation written in the normalized form:
\begin{equation}
\label{1}
u_t-6u^2u_x+u_{xxx}=0,
\end{equation}
where $u = u(x,t)$ is a real-valued function of two real-valued variables $(x,t)$. 
The defocusing mKdV equation (\ref{1}) is a canonical model which can be used to describe nonlinear phenomena in the physics of fluids and crystals, e.g. the propagation of internal waves \cite{Grimshaw,Pelinovsky}, meandering ocean currents \cite{Ralph}, or long waves in the chain of particles \cite{Vainshtein}. 

The purpose of this work is to obtain a new exact solution to the mKdV equation (\ref{1}) which describes the periodic interaction of a dark soliton and a traveling periodic wave. Due to periodicity of such interactions, we cast this solution as {\em the dark breather}. 

Breathers represent spatially localized, time-periodic wave patterns that persist in the nonlinear dynamics. They generalize solitons by incorporating an additional time scale associated with internal oscillations, and have been widely known in the context of integrable systems. Breathers of the KdV equations were studied in \cite{Bertola,HMP} after much earlier works \cite{Gesztesy,Kuznetsov}. Two families of bright (elevation) and dark (depression) breathers were constructed and compared with numerical and laboratory experiments of interactions between solitary waves and dispersive shock waves \cite{Cole,Mao}. Similar bright and dark breathers were obtained 
for another model of the Benjamin--Ono equation \cite{ChenPelin-BO}. 
Breathers of the focusing and defocusing NLS (nonlinear Schr\"{o}dinger) equations were obtained respectively in \cite{Feng} and \cite{Ling} after the previous works in \cite{CPnls,CPW} and \cite{Lou,Shin,Takahashi}. 

Breathers arising as a result of interactions of solitary waves and traveling periodic waves have been observed in various fluids  \cite{Gavr,Maiden,Sprenger,Sande}. For a better comparison with experiments, the breather solutions are needed to be constructed in the simplest form with all parameters  explicitly expressed in terms of the Jacobi elliptic functions. These representations are useful for the study of the physically observable parameters such as the speed, localization width, and the relative shift of dark solitons relative to the periodic background. 

The exact solutions for complicated wave interactions are available for the mKdV equation (\ref{1}) due to its integrability. This is expressed through the existence of the Lax system of two linear equations:
\begin{equation}\label{spectral}
\varphi_x = \begin{pmatrix}
i\zeta & u\\
u & -i\zeta
\end{pmatrix} \varphi 
\end{equation}
and 
\begin{equation}\label{time_evol}
\varphi_t = \begin{pmatrix}
4i\zeta^3+2i\zeta u^2 & 4\zeta^2u-2i\zeta u_x+ 2u^3-u_{xx}\\
4\zeta^2u+2i\zeta u_x+ 2u^3-u_{xx} & -4i\zeta^3-2i\zeta u^2
\end{pmatrix} \varphi, 
\end{equation}
where $\zeta$ is the $(x,t)$-independent spectral parameter and 
$\varphi = (p,q)^T$ is the corresponding eigenfunction. The mKdV equation (\ref{1}) appears as a compatibility condition $\varphi_{xt} = \varphi_{tx}$ of the Lax system (\ref{spectral}) and (\ref{time_evol}), see \cite{Ablowitz,ZS} 
for pioneering works. 

The spectral problem (\ref{spectral}) can be written 
as the classical eigenvalue problem:
\begin{equation}\label{eigen}
(\mathcal{L}-\zeta I)\varphi=0,\;\; \mathcal{L}=\begin{pmatrix}
-i\partial_x & iu\\
-i u & i \partial_x
\end{pmatrix}
\end{equation}
defined by a self-adjoint Dirac operator $\mathcal{L}$ with real $u$ in $L^2(\R)$. Spatially bounded solutions of the eigenvalue problem (\ref{eigen}) exist for admissible values $\zeta$ on the real line $\mathbb{R}$. The set of all admissible values of $\zeta$ is said to be the Lax spectrum associated with the given potential $u$.

Although the mKdV equation is related to the NLS equation because they share the same spectral problem (\ref{eigen}), there are differences between the complex-valued solutions to the NLS equation and the real-valued solutions 
to the mKdV equation. As a result, a general family of the traveling periodic wave solutions to the NLS equation from \cite{Ling} generates only one traveling periodic wave solution to the mKdV equation (\ref{1}):
\begin{equation}\label{sn_potential}
u(x,t)=\phi_0(x+c_0t),\quad \phi_0(x)=k\sn(x;k),\quad c_0=1+k^2,
\end{equation}
where $k \in (0,1)$ is the elliptic modulus. The snoidal solution  (\ref{sn_potential}) is expressed by the Jacobi elliptic function $\sn(x;k)$,  where the elliptic modulus $k\in(0,1)$ parametrizes the family. 
We note that $\phi_0(x) = 0$ as $k \to 0$ and 
$\phi_0(x) = \tanh(x)$ as $k \to 1$, where the latter solution is 
referred to as the kink of the mKdV equation (\ref{1}). The snoidal solution (\ref{sn_potential}) generates a more general family of the periodic solutions of the mKdV equation (\ref{1}) by means of the scaling transformation 
\begin{equation}
\label{scaling-trans}
u(x,t) = \alpha \phi_0(\alpha(x+ct)), \quad c = \alpha^2 c_0,
\end{equation}
where the parameter $\alpha > 0$ is arbitrary.

The family of the snoidal solutions  (\ref{sn_potential}) is related to two (symmetric) spectral bands of the eigenvalue problem (\ref{eigen}). This family is not equivalent to the general traveling periodic wave solution to the mKdV equation which has three spectral bands. In the case of the focusing mKdV equation, a similar constraint on the general family of elliptic solutions to the NLS equation \cite{Feng} generates only two particular (dnoidal and cnoidal) traveling periodic wave solutions to the mKdV equation \cite{Chen,Sun1,Sun2}, which are not equivalent to the general traveling periodic wave solutions explored in \cite{Chen-JNLS}. 

The scopes of our work are restricted to the dark breathers on the snoidal background (\ref{sn_potential}). We do so by using  the one-fold Darboux transformation and by transforming the solution to the simplest form due to properties of the Jacobi theta functions including the half-period and quarter-period transflations in the complex plane.  As the main novelty of our work, the obtained solutions did not appear in the previous publications on the defocusing NLS equations in  \cite{Lou,Shin,Takahashi}. The explicit expressions are used to draw information about the physically observable parameters such as the breather speed, localization width, and breather phase shift.

Breathers are constructed by picking an eigenvalue in one of the two (symmetric) spectral gaps of the Lax spectrum associated with the snoidal background (\ref{sn_potential}). This yields dark breathers, for which the dark soliton propagates on the snoidal background as a depression wave. These breathers are topological because they impart a phase
shift to the snoidal wave background. We show that dark breathers propagate faster than the snoidal wave background, while imparting a positive phase shift.

It remains open for further studies to obtain similar formulas for the general traveling periodic wave solutions related to three spectral gaps, in which case we anticipate co-existence of two breather solutions: dark breathers in the two (symmetric) gaps and kink breathers in the central gap. One technical problem which needs to be solved for construction of such breathers is to reformulate the solution $u(x,t)$ of the mKdV equation (\ref{1}) in the form of a quotient of a product of Jacobi theta functions (see the recent work in \cite{Geng}) and to obtain the explicit solutions of the Lax system in a similar form of a quotient of a product of Jacobi theta functions. This problem is left for future studies.

The organization of this paper is as follows. The main results featuring the closed-form expression for dark breathers are explained in Section \ref{sec-2}. 
Breather characteristics are described in Section \ref{sec-3}. The technical details of the proof are developed in Section \ref{sec-4}. Appendix \ref{sec-app} collects together some known relations between Jacobi elliptic functions. 

\vspace{0.2cm}

{\bf Acknowledgement.}  The authors thank M. A. Hoefer for many useful discussions and collaborations during the project. D. E. Pelinovsky 
is supported in part by the National Natural Science Foundation 
of China (No. 12371248).

\section{Main results}
\label{sec-2}

Throughout the work, we make use of two of the four Jacobi's theta functions \cite{Lawden}:
\begin{align*}
\theta_1(y) &= 2 \sum\limits_{n=1}^{\infty} (-1)^{n-1} q^{(n-\frac{1}{2})^2} \sin(2n-1) y
\end{align*}
and
\begin{align*}
\theta_4(y) &= 1 + 2 \sum\limits_{n=1}^{\infty} (-1)^{n} q^{n^2} \cos 2n y
\end{align*}
where $q := e^{-\frac{\pi K'(k)}{K(k)}}$ with $K(k)$ being the complete elliptic integral and $K'(k) = K(k')$ with $k' = \sqrt{1-k^2}$. It is well-known \cite{Lawden} that $K(k)$ is a quarter period and $i K'(k)$ is a half period of the Jacobi elliptic function $\sn(x;k)$ with the correspondence 
$$
y = \frac{\pi x}{2 K(k)}.
$$ 
For notational convenience, 
we use $H(x) = \theta_1(y)$, $\Theta(x) = \theta_4(y)$, and drop the dependence of $k \in (0,1)$ in the elliptic integrals and elliptic functions unless it creates a confusion. We also use $Z(x) = \Theta'(x)/\Theta(x)$. Using the relation \cite[(2.1.1)]{Lawden}:
\begin{align}
\label{sn-theta-quotient}
\sn(x;k) = \frac{H(x)}{\sqrt{k} \Theta(x)},  
\end{align}
we write the profile $\phi_0$ of the traveling wave (\ref{sn_potential}) in the equivalent form:
\begin{equation}
\label{sn-theta}
\phi_0(x) = \frac{H(K)}{\Theta(K)} \; \frac{H(x)}{\Theta(x)}.
\end{equation}

The following one-fold Darboux transformation allows us to obtain a new solution $\hat{u}$ of mKdV equation (\ref{1}) from the old solution $u$:
\begin{equation}\label{DT}
\hat{u}=u-\frac{4i\zeta p q}{p^2-q^2},
\end{equation}
where $\varphi = (p,q)^T$ is a particular nonzero solution of the linear systems  (\ref{spectral}) and (\ref{time_evol}), associated with the potential $u$ for a particular value of the spectral parameter $\zeta$. 
The validity of the one-fold Darboux 
transformation formula was recently confirmed in Appendix A in \cite{Chen} for the focusing mKdV equation. The new solution $\hat{u}$ in (\ref{DT}) is real if the complex-conjugate reduction $q = \bar{p}$ is satisfied on the solution $\varphi = (p,q)^T$ of the linear systems  (\ref{spectral}) and (\ref{time_evol}) with real $u$ and real $\zeta$. The solution is singular 
if $p = p(x,t)$ becomes real at a point $(x,t) \in \R\times \R$, and the main challenge of using the one-fold transformation (\ref{DT}) is to ensure 
that $\hat{u}(x,t)$ is bounded for every $(x,t) \in \R\times \R$.

We are now ready to present the main results of this work.

\begin{theorem}
	\label{theorem-1}
	Let $\alpha \in (0,K)$ be a free parameter in addition to $k \in (0,1)$ and define $\beta$ and $\gamma$ by 
	\begin{align}
	\label{beta}
	\beta := \left[ 1 - \frac{2 (1+k)^2 \sn^2(\alpha)}{(1 + k \sn^2(\alpha))^2} \right] \frac{\Theta(2 \alpha)}{\Theta(0)}, \qquad 
	\gamma := \frac{\Theta(2 \alpha)}{\Theta(0)}, 
	\end{align}
	The new solution of the mKdV equation (\ref{1}) for the dark breather is given in the form:
\begin{equation}
\label{new-solution}
\hat{u}(x,t) = \frac{H(K)}{\Theta(K)} \;
\frac{H(\xi + 2 \alpha) e^{-2\eta} + H(\xi - 2 \alpha) e^{2\eta}  + 2 \beta H(\xi) }{\Theta(\xi + 2 \alpha) e^{-2\eta} + \Theta(\xi - 2 \alpha) e^{2\eta}  + 2 \gamma \Theta(\xi)},
\end{equation}
where $\xi = x + c_0 t$ with $c_0 = 1+k^2$ and $\eta = \kappa (x + ct + x_0)$ with 
\begin{equation}
\label{kappa}
\kappa :=  Z(\alpha) + \frac{k\sn(\alpha)\cn(\alpha)\dn(\alpha)}{1+k\sn^2(\alpha)} > 0,
\end{equation}
\begin{equation}
\label{speed}
c := c_0 +\frac{2k(1+k)^2[1-k\sn^2(\alpha)] \sn(\alpha)\dn(\alpha)\cn(\alpha)}{[Z(\alpha)[1+k\sn^2(\alpha)]+k\sn(\alpha)\dn(\alpha)\cn(\alpha)][1+k\sn^2(\alpha)]^2} > c_0.
\end{equation}
and arbitary $x_0 \in \mathbb{R}$.
\end{theorem}

\begin{remark}
Since $\Theta(x) > 0$ for every $x \in \mathbb{R} $, the new solution $\hat{u}(x,t)$ is non-singular for every $(x,t) \in \mathbb{R} \times \mathbb{R}$.
\end{remark}

\begin{remark}
The value of $\alpha$ is uniquely defined by the spectral parameter $\zeta$ from the characteristic equation:
\begin{equation}
\label{char}
\zeta =\frac{1+k}{2} \; \frac{1-k\sn^2(\alpha)}{1+k\sn^2(\alpha)}
\end{equation}
so that if $\alpha \in (0,K)$, then $\zeta \in (\zeta_-,\zeta_+)$ 
with $\zeta_{\pm} := \frac{1}{2} (1 \pm k)$. Intervals 
$(-\zeta_+,-\zeta_-)$ and $(\zeta_-,\zeta_+)$ are two (symmetric) gaps 
in the Lax spectrum in the spectral problem (\ref{eigen}) associated with the potential $u(x,t) = \phi_0(x+c_0t)$ with the profile $\phi_0$ given by (\ref{sn-theta}).
\end{remark}

\begin{remark}
	The value $\alpha \in (0,K)$ determines the phase shift of the solitary wave propagating across the snoidal background because 
	$$
	\lim_{\eta \to \pm \infty} \hat{u}(x,t) = k \sn(\xi \mp 2 \alpha).
	$$
	Since the period of $\sn(\xi)$ is $4K$, a suitably normalized phase shift can be defined by 
\begin{equation}
\label{Delta}
\Delta := \frac{2\pi(4\alpha)}{4K}=\frac{2\pi\alpha}{K}\in (0,2\pi).
\end{equation}
The inverse localization width of the solitary wave is defined by $\kappa$ in (\ref{kappa}) and its velocity is defined by  $c$ in (\ref{speed}). 
\end{remark}

Figure \ref{kink} shows the spatiotemporal evolution of a dark breather on the snoidal wave background given by the new solution in Theorem \ref{theorem-1}. The breather travels faster than the periodic wave (both waves travel to the left) and imparts a phase shift. 

\begin{figure}[htb!]
	\centering
	\includegraphics[width=7.5cm,height=7cm]{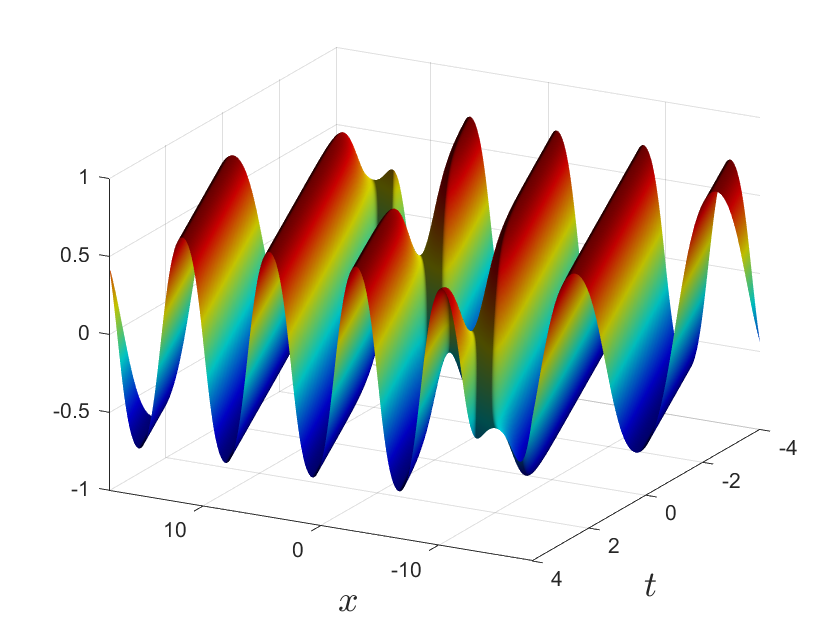}
	\includegraphics[width=7cm,height=7cm]{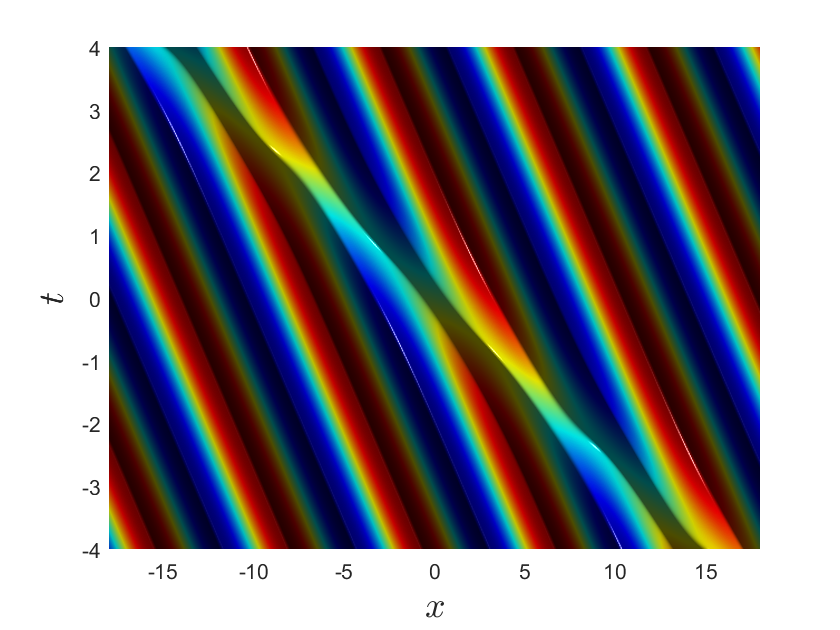}
	\caption{Breather on the periodic wave background for $\alpha=0.3K(k)$, $k =0.7$, and $x_0=0$.}
	\label{kink}
\end{figure}

\begin{corollary}
	\label{theorem-2}
In the limit $k \to 1$, the family of dark breathers of Theorem \ref{theorem-1} generates a two-soliton solution of the mKdV equation (\ref{1}) in the form:
\begin{equation}
\label{two-solitons}
\hat{u}(x,t) = 
\frac{\sinh(\xi+2\alpha)  e^{-2\eta}  +  \sinh(\xi-2\alpha) e^{2\eta} + 2 
	\sinh(\xi) (1 - \sinh^2(2\alpha)) {\rm sech}(2\alpha)}{\cosh(\xi+2\alpha)  e^{-2\eta}  +  \cosh(\xi-2\alpha) e^{2\eta} + 2 
	\cosh(\xi) \cosh(2\alpha)},
\end{equation}
where
\begin{align*}
\xi &= x + 2 t, \\
\eta &= \tanh(2\alpha) [x  + 2 t +  4 t \; {\rm sech}^2(2\alpha) + x_0],
\end{align*}
with $x_0 \in \mathbb{R}$ and $\alpha \in (0,\infty)$ being free parameters. 
\end{corollary}

\begin{remark}
	The parametrization formula (\ref{char}) in the limit $k \to 1$ becomes 
	$\zeta = {\rm sech}(2\alpha)$ so that if the first soliton has speed $2$, then the second soliton has speed $2 + 4 \zeta^2$ with $\zeta \in (0,1)$.
\end{remark}

Figure \ref{two-soliton} shows the spatiotemporal evolution of the two-soliton solution with the profile $\hat{u}$ for $\alpha = 0.6$ (which corresponds to the eigenvalue $\zeta = 0.552$). The two solitons propagate with different speeds, collide, and scatter after some interaction.

\begin{figure}[htb!]
	\centering
	\includegraphics[width=7.5cm,height=7cm]{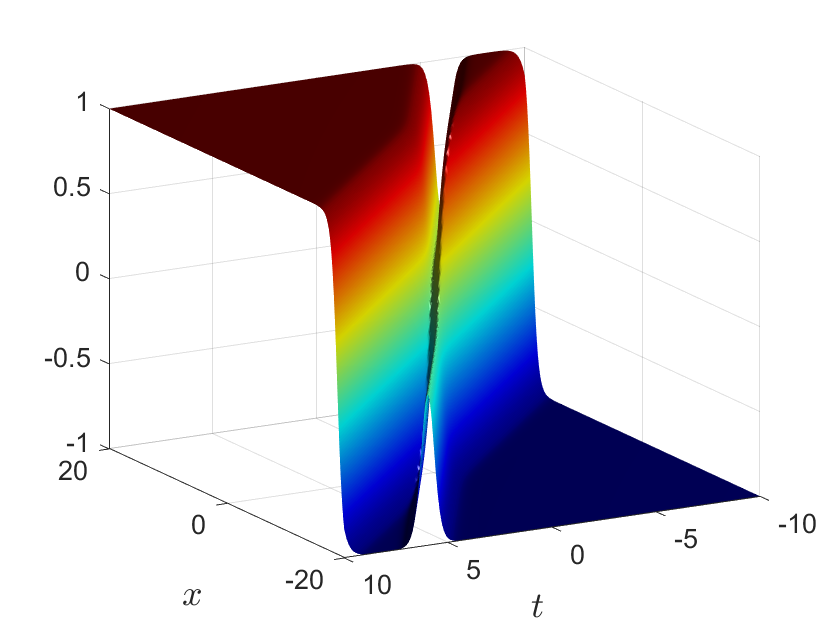}
	\includegraphics[width=7cm,height=7cm]{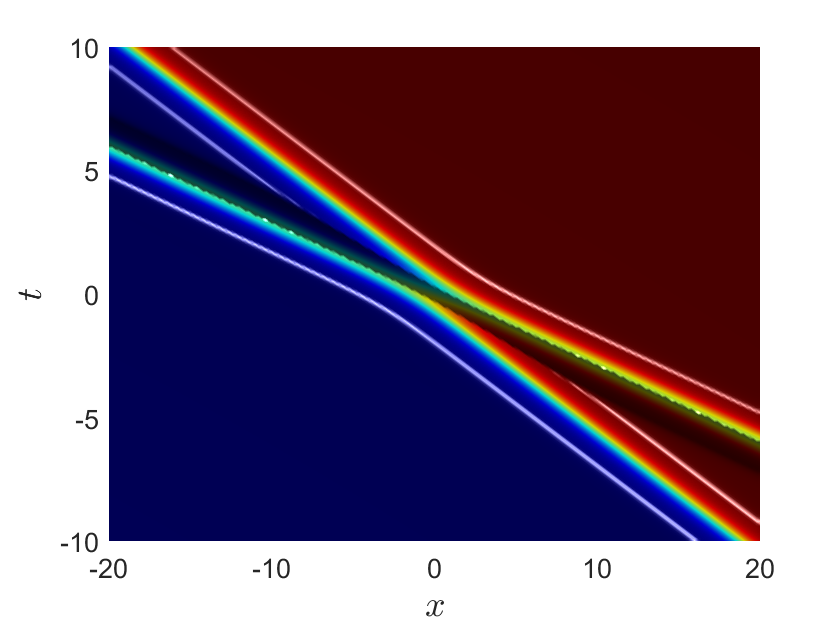}
	\caption{Two-soliton solution of the mKdV equation (\ref{1}) for $\alpha= 0.6$ and $x_0 = 0$.}
	\label{two-soliton}
\end{figure}

\section{Properties of the dark breather}
\label{sec-3}

Here we explore the characteristic properties of the breather solutions given by  Theorem \ref{theorem-1}. In particular, we analyze the breather phase shift $\Delta$, the localization parameter $\kappa$, 
and the breather speed $c$. All parameters are characterized by $\alpha \in (0,K)$ which is uniquely determined via the characteristic equation (\ref{char}) by the value  of the spectral parameter $\zeta$ in the spectral gap $(\zeta_-,\zeta_+)$.

The following two lemmas give monotonicity of the mapping $\zeta \to \Delta$ and the existence of a single maximum in the mapping $\zeta \to \kappa$.

\begin{lemma}
    \label{lem-1}
    The phase shift $\Delta$ is a monotonically decreasing function of $\zeta$ in $(\zeta_-,\zeta_+)$ with $\Delta(\zeta_-) = 2\pi$ and $\Delta(\zeta_+) = 0$.
\end{lemma}

\begin{proof}
	It follows from the characteristic equation (\ref{char}) that 
	\begin{equation}\label{sin_phi}
	\sn(\alpha) = \frac{\sqrt{1+k-2\zeta}}{\sqrt{k(1+k+2\zeta)}} = \sin(\varphi_{\alpha}),
	\end{equation}
	where $\varphi_{\alpha}$ is defined by $\alpha$ through the incomplete elliptic integral of the first kind as $\alpha = F(\varphi_{\alpha},k)$. 
	If $\alpha \in (0,K)$, then $\varphi_{\alpha} \in \left(0,\frac{\pi}{2}\right)$. Using (\ref{Delta}) together with (\ref{sin_phi}), we obtain 
$$
\partial_{\zeta} \Delta = \frac{2\pi}{K} \partial_{\varphi_{\alpha}} F(\varphi_{\alpha},k)\partial_{\zeta}\varphi_{\alpha}.
$$
Differentiating (\ref{sin_phi}) in $\zeta$ yields $\partial_{\zeta} \varphi_{\alpha} < 0$ for $\varphi_{\alpha} \in \left(0,\frac{\pi}{2}\right)$. 
On the other hand, it follows from the definition of the incomplete elliptic integral that $\partial_{\varphi} F(\varphi,k) > 0$. Hence, we have 
$\partial_{\zeta} \Delta < 0$. Since $\alpha = 0$ at $\zeta = \zeta_+$ 
and $\alpha = K$ at $\zeta = \zeta_-$, it follows from (\ref{Delta}) 
that $\Delta(\zeta_-) = 2\pi$ and $\Delta(\zeta_+) = 0$. 
\end{proof}

\begin{lemma}
	\label{lem-2}
	The localization parameter $\kappa$ admits the only extremal (maximum) point in $(\zeta_-,\zeta_+)$ at 
	\begin{equation}
	\label{maximum}
	\zeta_0 = \sqrt{\frac{E}{2K}-\frac{(1-k^2)}{4}},
	\end{equation}
	where $E$ is a complete elliptic integral of the second kind.
\end{lemma}

\begin{proof}
	By Lemma \ref{lem-1}, the mapping $\zeta \to \Delta$ is monotone, 
	where $\Delta = 2 \pi \alpha/K$ in (\ref{Delta}).  Hence, we can check the mapping $\alpha \to \kappa$ instead of $\zeta \to \kappa$. Computing the derivative of (\ref{kappa}) in $\alpha$, we obtain a critical point of the mapping $\alpha \to \kappa$ from roots of the transcendal equation:
\begin{align}
\notag & 1-\frac{E}{K}-k^2\sn^2(\alpha)-\frac{2k^2\sn^2(\alpha)\cn^2(\alpha)\dn^2(\alpha)}{(1+k\sn^2(\alpha))^2} \\
& \qquad  + k \frac{\cn^2(\alpha)\dn^2(\alpha)-\sn^2(\alpha)\dn^2(\alpha)-k^2\sn^2(\alpha)\cn^2(\alpha)}{1+k\sn^2(\alpha)}=0,
\label{transc}
\end{align}
where we have used the formula $Z'(\alpha) = 1 - k^2\sn^2(\alpha) - E/K$. We can use 
\begin{equation}
\label{appendix-extra}
\cn^2(\alpha) \dn^2(\alpha) + (1+k)^2 \sn^2(\alpha) = [1+k \sn^2(\alpha)]^2
\end{equation} 
and the fundamental relations 
\begin{equation}
\label{fund-rel}
\cn^2(\alpha) = 1 - \sn^2(\alpha), \qquad \dn^2(\alpha)  = 1 - k^2 \sn^2(\alpha)
\end{equation} 
to obtain 
\begin{align*}
& \;\; 1-k^2\sn^2(\alpha)-\frac{2k^2\sn^2(\alpha)\cn^2(\alpha)\dn^2(\alpha)}{(1+k\sn^2(\alpha))^2} \\
& \qquad \qquad  + k \frac{\cn^2(\alpha)\dn^2(\alpha)-\sn^2(\alpha)\dn^2(\alpha)-k^2\sn^2(\alpha)\cn^2(\alpha)}{1+k\sn^2(\alpha)} \\
& = 1 - 3 k^2 \sn^2(\alpha) + \frac{2 k^2 (1+k)^2 \sn^4(\alpha)}{(1+k\sn^2(\alpha))^2} + k \frac{1 - 2(1+k^2) \sn^2(\alpha) + 3 k^2 \sn^4(\alpha)}{1 + k \sn^2(\alpha)} \\
& =  1 + \frac{2 k^2 (1+k)^2 \sn^4(\alpha)}{(1+k\sn^2(\alpha))^2} + k \frac{1 - 2(1+k^2) \sn^2(\alpha) - 3k \sn^2(\alpha)}{1 + k \sn^2(\alpha)} \\
& = 1 + k \frac{(1+k \sn^2(\alpha))^2 - 2 (1+k)^2 \sn^2(\alpha)}{(1+k\sn^2(\alpha))^2} \\
& = (1+k) \left[ 1 - \frac{2k (1+k) \sn^2(\alpha)}{(1+k\sn^2(\alpha))^2} \right].
\end{align*}
The transcendental equation (\ref{transc}) is rewritten in the form 
\begin{align*}
\frac{E}{K} &= \frac{(1-k^2)}{2}+\frac{(1+k)^2}{2}\frac{(1-k\sn^2(\alpha))^2}{(1+k\sn^2(\alpha))^2} \\
&= \frac{(1-k^2)}{2}+2\zeta^2,
\end{align*}
which yields (\ref{maximum}). It follows from (\ref{kappa}) that 
$\kappa > 0$ for $\alpha \in (0,K)$ with $\kappa \to 0$ as $\alpha \to 0$ and $\alpha \to K$. Since there is only one critical point of $\kappa$ for positive $\zeta$, the mapping $\zeta \to \kappa$ is monotonically increasing for $\zeta \in (\zeta_-,\zeta_0)$ and monotonically decreasing for $\zeta \in (\zeta_0,\zeta_+)$ with the global maximum in $(\zeta_-,\zeta_+)$ at $\zeta_0$.
\end{proof}

Figure \ref{sol_prop} plots $\Delta$, $\kappa$, and $c$ as a function of the spectral parameter $\zeta$ in the spectral gap $[\zeta_-,\zeta_+]$. The band edges $\zeta_-$ and $\zeta_+$ are shown by the vertical dashed lines. The phase shift $\Delta$ is monotonically decreasing between the band edges in agreement with Lemma \ref{lem-1}. The inverse width $\kappa$ has
a single maximum and vanishes at the band edges in agreement with Lemma \ref{lem-2}. The breather speed $c$ is monotonically increasing and satisfy $c > c_0$. 

\begin{figure}[htb!]
	\centering
	\includegraphics[width=5cm,height=5cm]{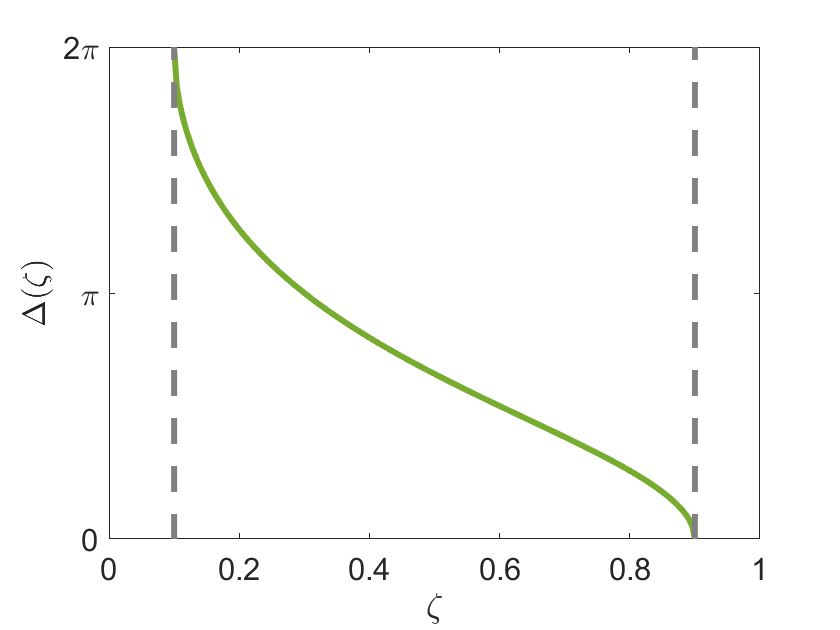}
	\includegraphics[width=5cm,height=5cm]{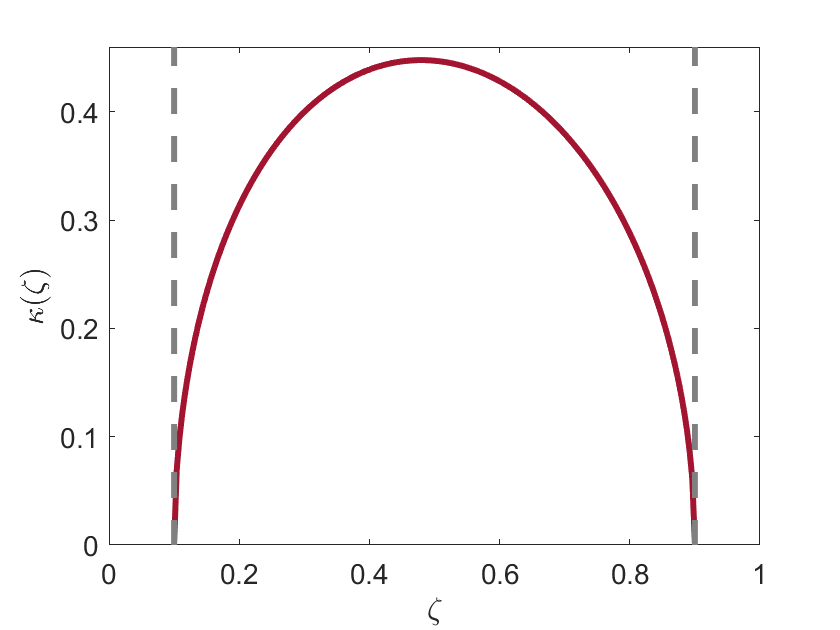}
	\includegraphics[width=5cm,height=5cm]{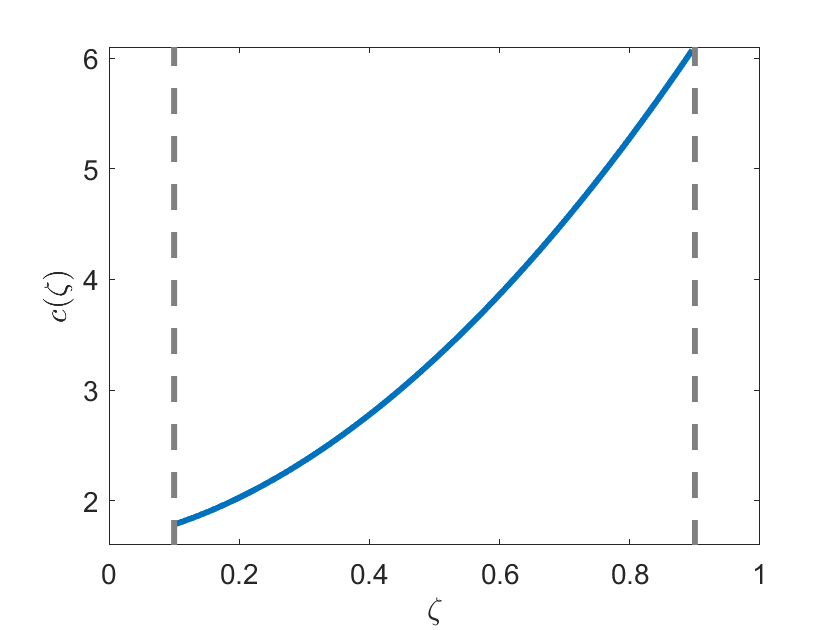}
	\caption{Normalized phase shift $\Delta$ (left), inverse width $\kappa$ (middle), and breather speed $c$
		(right) versus $\zeta \in (\zeta_-, \zeta_+)$ for $k = 0.8$. The band edges $\zeta_-$ and $\zeta_+$ 
		are shown by the vertical dashed lines.}
	\label{sol_prop}
\end{figure}

\begin{figure}[htb!]
	\centering
	\includegraphics[width=\linewidth]{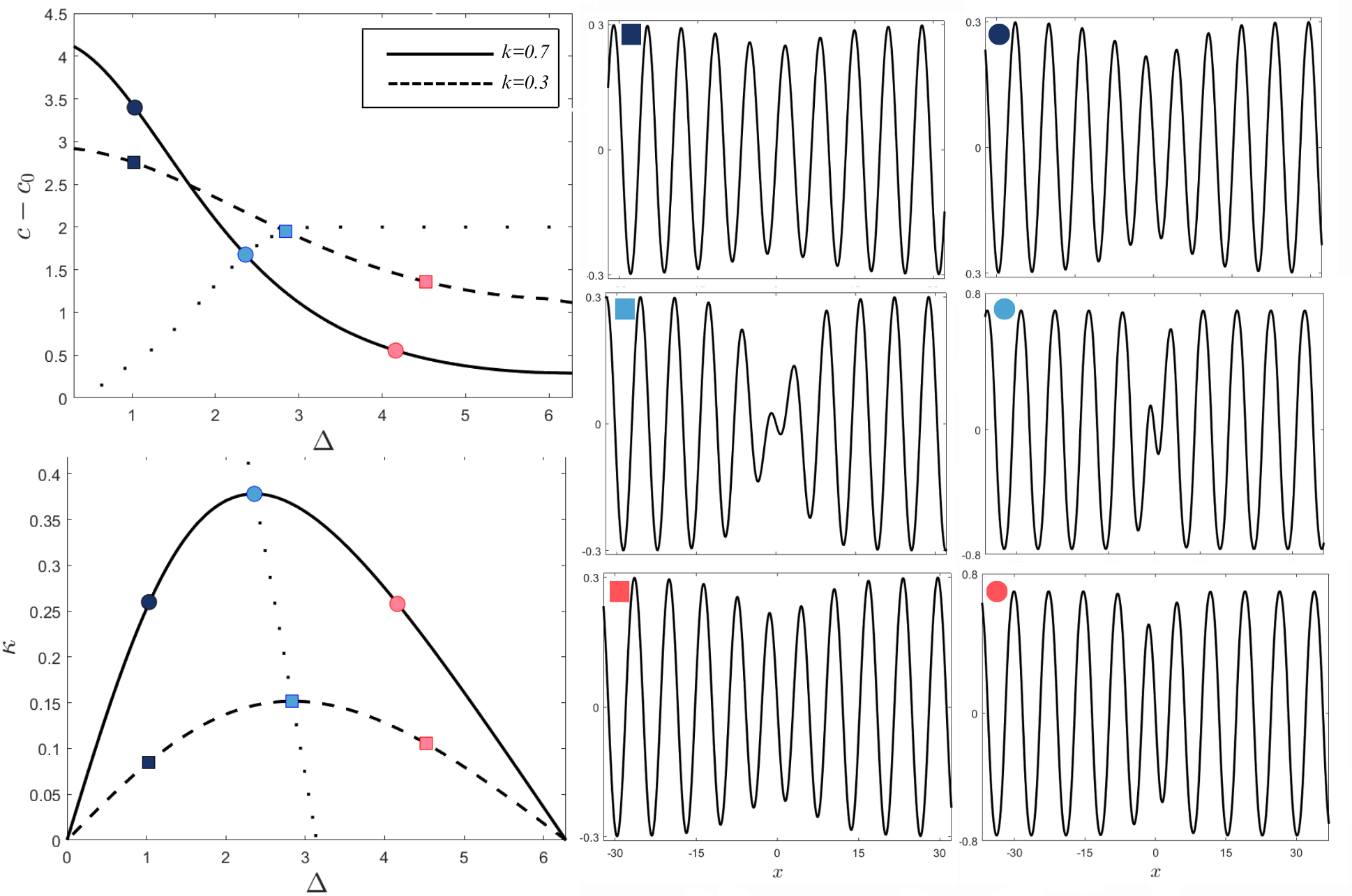}
	\caption{Left: plots of $c-c_0$ and $\kappa$ versus $\Delta$ for several values of $k$. Right: representative dark breather solutions. Representative solutions are marked on
		the left panel with a unique colored symbol. The dotted curve on the left panels corresponds to points of maximal $\kappa$ parameterized by $k$.}
	\label{charac}
\end{figure}

Figure \ref{charac} shows the profiles of the family of breathers for two values of $k$. The periodic wave background is close to a sinusoidal wave for smaller values of $k$ and is close to a kink as $k \to 1$ on each period $[0,4 K]$. The shift parameter $\alpha$ determines the breather localization relative to the periodic wave background. When $\alpha \to 0$ and $\alpha \to K$, the breather represents a slowly modulated wave over many periods since the inverse width parameter $\kappa$ becomes smaller. When $\alpha \to \alpha_{\rm max}$ given by the root of (\ref{transc}),  the breather has the narrowest (strongest) modulation of the cnoidal wave. The dotted curves in the left panels show the graphs of
$$
\{ (\Delta_{\rm max},c_{\rm max}-c_0), \quad  k\in(0,1)\} \quad \mbox{\rm and} \quad \{  (\Delta_{\rm max},\kappa_{\rm max}), \quad  k\in(0,1)\} 
$$
where $\Delta_{\rm max}$, $c_{\rm max}$, and $\kappa_{\rm max}$ are computed at $\alpha = \alpha_{\rm max}$.

\section{Proof of the main results}
\label{sec-4}

The starting point for the proof of Theorem \ref{theorem-1} is the exact solution of the following system of differential equations:
\begin{equation}
\label{ZS-norm}
\left\{ \begin{array}{l} 
p'(x) = i\zeta p(x) + \phi_0(x) q(x), \\
q'(x) = -i\zeta q(x) + \phi_0(x) p(x), 
\end{array} \right.
\end{equation}
where $\phi_0(x) = k \sn(x)$, $x$ stands for $x + c_0t$, and $\varphi = (p,q)^{T}$ gives a solution of the spectral problem (\ref{spectral}) for the normalized wave (\ref{sn_potential}) after the translation. 

We take for granted (see \cite{Takahashi} based on earlier works \cite{Dunke,Shin,Smirnov}) that the system (\ref{ZS-norm}) is satisfied by the following explicit functions
\begin{equation}
\label{p}
p(x;z) = e^{s(z) x} e^{- \frac{i \pi x}{4K}}
\frac{H(x-iz)}{\Theta(x) \Theta(i z)}, \qquad 
q(x;z) = e^{s(z) x} e^{- \frac{i \pi x}{4K}}
\frac{\Theta(x-iz)}{\Theta(x) H(i z)},
\end{equation}
where $s(z)$ is defined by 
\begin{equation}
\label{s-definition}
s(z) = \frac{1}{2} Z(iz) - \frac{1}{2} Z(iz'), \quad z' = K'-z.
\end{equation}
The spectral parameter $z \in \mathbb{C}$ is related to the spectral parameter $\zeta \in \mathbb{R}$ of the linear system (\ref{ZS-norm}) by the characteristic equation
\begin{equation}
\label{zeta-definition}
\zeta(z) = \frac{1}{2}\dn(iz)\dn(iz'). 
\end{equation}
The second linearly independent solution of the system (\ref{ZS-norm}) is obtained from (\ref{p}) by replacing $z$ with $z'$ and vice versa.

\subsection{Lax spectrum for the snoidal potential}

The following result is based on the study of the characteristic equations 
(\ref{s-definition}) and (\ref{zeta-definition}). See Figure \ref{fig:path} 
for illustration of the Lax spectrum and the corresponding values of the shift parameter $z$ in (\ref{zeta-definition}).

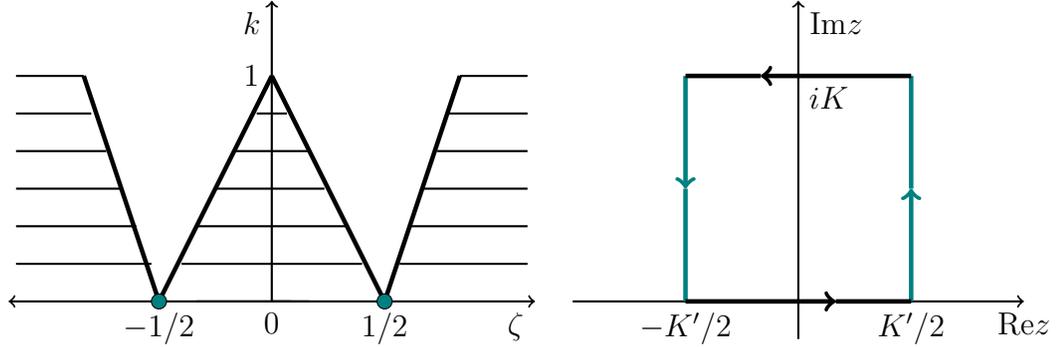
\begin{figure}[h!]
	\centering
	\begin{tikzpicture}
	\draw[thick] (-4.5,0) -- (-2,0) node[anchor=north west] {};
	\draw[thick,->] (-2,0) -- (-0,0) node[anchor=north east] {$\zeta$};
	\draw[thick,->] (-3.5,0) -- (-3.5,4) node[anchor=north east] {$k$};
	\draw[thick,->] (-5,0) -- (-7,0) node[anchor=north east] {};
	\draw[line width=0.6mm,black] (-5,0) -- (-3.5,3) node[anchor=east] {1};
	\draw[thick] (-4.7,0.5) -- (-2.3,0.5) node[anchor=east] {};
	\draw[thick] (-4.5,1) -- (-2.5,1) node[anchor=east] {};
	\draw[thick] (-4.25,1.5) -- (-2.8,1.5) node[anchor=east] {};
	\draw[thick] (-4,2) -- (-3,2) node[anchor=east] {};
	\draw[thick] (-3.7,2.5) -- (-3.3,2.5) node[anchor=east] {};
	\draw[thick] (-1.8,0.5) -- (-0.1,0.5) node[anchor=east] {};
	\draw[thick] (-1.65,1) -- (-0.1,1) node[anchor=east] {};
	\draw[thick] (-1.5,1.5) -- (-0.1,1.5) node[anchor=east] {};
	\draw[thick] (-1.3,2) -- (-0.1,2) node[anchor=east] {};
	\draw[thick] (-1.15,2.5) -- (-0.1,2.5) node[anchor=east] {};
	\draw[thick] (-1,3) -- (-0.1,3) node[anchor=east] {};
	\draw[thick] (-5.1,0.5) -- (-6.9,0.5) node[anchor=east] {};
	\draw[thick] (-5.3,1) -- (-6.9,1) node[anchor=east] {};
	\draw[thick] (-5.5,1.5) -- (-6.9,1.5) node[anchor=east] {};
	\draw[thick] (-5.7,2) -- (-6.9,2) node[anchor=east] {};
	\draw[thick] (-5.9,2.5) -- (-6.9,2.5) node[anchor=east] {};
	\draw[thick] (-6,3) -- (-6.9,3) node[anchor=east] {};
	\draw[line width=0.6mm,black] (-2,0) -- (-3.5,3) node[anchor=north east] {};
	\draw[line width=0.6mm,black] (-2,0) -- (-1,3);
	\draw[line width=0.6mm,black] (-5,0) -- (-6,3);
	\draw[thick] (-3,0) -- (-5,0);
	\draw[thick] (-3.25,0) -- (-3.5,0) node[anchor=north] {$0$};
	\filldraw[fill=teal, draw=black] (-2,0) circle (1mm)node[anchor=north] {$1/2$}; 
	\filldraw[fill=teal, draw=black] (-5,0) circle (1mm) node[anchor=north] {$-1/2$};
	
	\draw[thick,->] (0.5,0) -- (6.5,0) node[anchor=north] {Re$z$};
	\draw[thick,->] (3.5,-0.5) -- (3.5,4) node[anchor=north west] {Im$z$};
	\draw[line width=0.6mm,teal,->] (5,0) -- (5,1.5) node[anchor=north west] {};
	\draw[line width=0.6mm,teal] (5,1.5) -- (5,3) node[anchor=north west] {};
	\draw[line width=0.6mm,teal] (2,0) -- (2,1.5) node[anchor=north west] {};
	\draw[line width=0.6mm,teal,<-] (2,1.5) -- (2,3) node[anchor=north west] {};
	\draw[line width=0.6mm,<-] (3,3) -- (5,3) node[anchor=north west] {};
	\draw[line width=0.6mm] (2,3) -- (3,3) node[anchor=north west] {};
	\draw[line width=0.6mm,->] (2,0) -- (4,0) node[anchor=north west] {};
	\draw[line width=0.6mm] (4,0) -- (5,0) node[anchor=north west] {};
	\draw (3.5,3) node[anchor=north west] {$iK$};
	\draw (5,0) node[anchor=north] {$K'/2$};
	\draw (2,0) node[anchor=north] {$-K'/2$};
	\end{tikzpicture}
	\caption{Left: Lax spectrum for $\phi_0(x) = k \sn(x)$ with $k\in(0,1)$. Right: The path of $z$ in the complex plane for real values of $\zeta$. }
	\label{fig:path}
\end{figure}

\begin{proposition}
	\label{prop-1}
Lax spectrum associated with the snoidal potential $\phi_0$ is located on 
$$
(-\infty,-\zeta_+] \cup [-\zeta_-,\zeta_-] \cup [\zeta_+,\infty)
$$
and the two band gaps are located on $(-\zeta_+,-\zeta_-)$ and $(\zeta_-,\zeta_+)$, where $\zeta_{\pm} := \frac{1}{2} (1 \pm k)$.	
\end{proposition}

\begin{proof}
	We have $\zeta(z) \in \R$ if either $z \in \mathbb{R} + i m K$ or
	$z \in i \mathbb{R} + \frac{1}{2} (2m+1) K'$, where $m \in \mathbb{Z}$. The former follows from the characteristic equation (\ref{zeta-definition}), the translation formulas (\ref{appendix-4}), and the reflection formulas (\ref{appendix-1}). The latter follows from the characteristic equation (\ref{zeta-definition}), the translation formulas (\ref{appendix-2}), 
	and the addition formulas (\ref{appendix-quarter-period}).

When $z$ traverses along the rectange in the complex plane shown in Figure \ref{fig:path} (right), the values of $\zeta$ in (\ref{zeta-definition}) change from $+\infty$ to $-\infty$, where $\zeta = \pm \infty$ corresponds to $z = 0$. The four corner points of the rectange in the $z$ plane correspond to 
\begin{align*}
\pm \zeta_+ := \zeta\left(\pm \frac{1}{2} K'\right) = \pm \frac{1}{2} (1+k) \quad 
\mbox{\rm and} \quad 
\pm \zeta_- := \zeta\left(\pm \frac{1}{2} K' + iK \right) = \pm \frac{1}{2} (1-k).
\end{align*}
The values of $s(z)$ in (\ref{s-definition}) are purely imaginary if $z \in \mathbb{R} + i m K$ and purely real if $z \in i \mathbb{R} + \frac{1}{2} (2m+1) K'$ for an integer $m$. Lax spectrum of the spectral problem 
(\ref{ZS-norm}) is defined by bounded solutions (\ref{p}) in $x$ 
which only exist if $s(z) \in i \mathbb{R}$. Thus, Lax spectrum corresponds to 
\begin{itemize}
	\item $[\zeta_+,\infty)$ for $z \in (0,\frac{1}{2} K']$, 
	\item $[-\zeta_-,\zeta_-]$ for $z \in [-\frac{1}{2} K',\frac{1}{2} K'] + i K$, 
	\item $(-\infty,-\zeta_+]$ for $z \in [-\frac{1}{2} K',0)$
\end{itemize} 
with two (symmetric) band gaps $(-\zeta_+,-\zeta_-)$ and $(\zeta_-,\zeta_+)$
for $z \in \pm \frac{1}{2} K' + i [0,K]$.
\end{proof}

\begin{remark}
Figure \ref{fig:path} (left) shows the Lax spectrum described in Proposition \ref{prop-1} for different values of $k \in (0,1)$. As $k \to 0$, the Lax spectrum transforms to $(-\infty,\infty)$. 
As $k \to 1$, the Lax spectrum transforms to $(-\infty,-1] \cup \{0\} \cup [1,\infty)$, where $0$ is the isolated eigenvalue of the eigenvalue problem (\ref{eigen}) for the black soliton \cite{ZS}. 
\end{remark}

\subsection{Parameterization in the spectral gap $(\zeta_-,\zeta_+)$}

The spectral gap $(\zeta_-,\zeta_+)$ corresponds to the vertical segment with 
${\rm Re}(z) = {\rm Re}(z') = \frac{1}{2} K'$, for which it is natural 
to parameterize $z$ by using 
\begin{equation}\label{z}
z = \frac{1}{2}K' + i\alpha, \quad \alpha\in[0,K].
\end{equation}
Since the dark breathers on the snoidal background are constructed by 
using the one-fold Darboux transformation (\ref{DT}) with $\zeta$ selected in the band gap, we shall give the explicit expressions for $s(z)$ and $\zeta(z)$
in (\ref{s-definition}) and (\ref{zeta-definition}) by using (\ref{z}).

\begin{proposition}
	\label{prop-2}
Let $z$ be given by the parameterization (\ref{z}). Then we have 
\begin{equation}
\label{char-prop}
\zeta =\frac{1+k}{2} \; \frac{1-k\sn^2(\alpha)}{1+k\sn^2(\alpha)}
\end{equation}
and
\begin{equation}
\label{s-expression}
s = -Z(\alpha) - \frac{k\sn(\alpha)\cn(\alpha)\dn(\alpha)}{1+k\sn^2(\alpha)}.
\end{equation}
\end{proposition}

\begin{proof}
	By using (\ref{zeta-definition}),  (\ref{z}), and (\ref{appendix-quarter-period}), we obtain
	\begin{align*}
	\zeta &= \frac{1}{2} \dn\left(\frac{i}{2}K'-\alpha \right) 
	\dn\left(\frac{i}{2}K'+\alpha \right) \\
	&= \frac{1+k}{2} \; \frac{\dn^2(\alpha) + k^2 \sn^2(\alpha) \cn^2(\alpha)}{[1 + k \sn^2(\alpha)]^2},
	\end{align*}
which yields (\ref{char-prop}). 

By using (\ref{s-definition}), (\ref{z}),  (\ref{appendix-quarter-period}), and (\ref{appendix-6}), we obtain
\begin{align*}
s &= \frac{1}{2} Z\left(\frac{iK'}{2} - \alpha \right) - \frac{1}{2} Z\left(\frac{iK'}{2} + \alpha\right) \\
&= -Z(\alpha) + \frac{1}{2} k^2 \sn\left(\frac{iK'}{2}\right) \sn(\alpha)\left[ \sn\left(\frac{iK'}{2} + \alpha \right)+ \sn\left(\frac{iK'}{2} - \alpha \right) \right], 
\end{align*}
which yields (\ref{s-expression}).
\end{proof}

\begin{remark}
	In agreement with Proposition \ref{prop-1}, it follows from (\ref{char-prop}) that $\zeta = \zeta_+$ at $\alpha=0$ and $\zeta = \zeta_-$ at $\alpha=K$, whereas it follows from (\ref{s-expression}) that  the exponent $e^{s(z) x}$ in (\ref{p}) is purely real for $\alpha \in (0,K)$.
\end{remark}

\subsection{Time evolution of eigenfunctions in the spectral gap}

The time evolution of eigenfunctions (\ref{p}) along the linear flow (\ref{time_evol}) is obtained by changing $x$ to $x + c_0t$ and by multiplication of the eigenfunctions 
by $e^{\omega(z) t}$ with $\omega(z)$ to be determined:
\begin{equation}
\label{time-evol}
\varphi(x,t) = e^{\omega(z) t} \left[ \begin{array}{l} p(x+c_0t;z) \\ q(x+c_0t;z) \end{array} \right].
\end{equation}
Substituting (\ref{time-evol}) into (\ref{time_evol}) 
and using (\ref{ZS-norm}) yields the algebraic system: 
\begin{equation}
\label{ZS-time}
\left\{ \begin{array}{l} 
\omega p = i\zeta (4 \zeta^2 + 2 \phi_0^2 - 1 - k^2) p + 
(4 \zeta^2 \phi_0 - 2 i \zeta \phi_0' + 2 \phi_0^3 - \phi_0'' - (1+k^2) \phi_0) q, \\
\omega q = (4\zeta^2 \phi_0 + 2i \zeta \phi_0' + 2 \phi_0^3 - \phi_0'' - (1+k^2) \phi_0) p - i\zeta (4 \zeta^2 + 2 \phi_0^2 - 1 - k^2) q,
\end{array} \right.
\end{equation}
where $\phi_0(x) = k \sn(x)$. 

\begin{proposition}
	\label{prop-3}
	The value of $\omega$ in the system (\ref{ZS-time}) is determined by $\zeta \in \mathbb{R}$ from 
\begin{equation}
\label{omega-definition}
\omega^2 = -16 \zeta^2 P(\zeta), \quad P(\zeta) := \zeta^4 - \frac{1}{2} (1+k^2)\zeta^2+\frac{1}{16}(1-k^2)^2,
\end{equation}
with $P(\zeta) = (\zeta^2 - \zeta_+^2) (\zeta^2 - \zeta_-^2)$ with 
$\zeta_{\pm} = \frac{1}{2} (1 \pm k)$.	
\end{proposition}

\begin{proof}
Since  (\ref{ZS-time}) is a linear algebraic system, we obtain 
the values of $\omega$ from the determinant equation
\begin{align*}
\det &\begin{pmatrix}
i\zeta(4\zeta^2+2\phi_0^2-1-k^2)-\omega& 4 \zeta^2 \phi_0-2i\zeta \phi_0'+2\phi_0^3-\phi_0''-(1+k^2)\phi_0 \\ 
4\zeta^2\phi_0+2i\zeta\phi_0'+2\phi_0^3-\phi_0''-(1+k^2)\phi_0 &   -i\zeta(4\zeta^2+2\phi_0^2-1-k^2)-\omega
\end{pmatrix}\\
&= \omega^2 +  \zeta^2(4\zeta^2+2\phi_0^2-1-k^2)^2-(4\zeta^2\phi_0-2i\zeta\phi_0'+2\phi_0^3-\phi_0''-(1+k^2)\phi_0)\\
& \quad \times(4\zeta^2\phi_0+2i\zeta \phi_0'+2\phi_0^3-\phi_0''-(1+k^2)\phi_0) = 0.
\end{align*}
The profile $\phi_0(x) = k \sn(x)$ satisfies the second-order differential equation:
\begin{equation}
\label{second}
\phi_0'' - 2 \phi_0^3 + c_0 \phi_0 = 0, \quad c_0 = 1 + k^2.
\end{equation}
Integration of (\ref{second}) yields the first-order invariant 
\begin{equation}
\label{third}
( \phi_0')^2 - \phi_0^4 + c_0 \phi_0^2 = d_0, \quad d_0 = k^2.
\end{equation}
Using (\ref{second}) and (\ref{third}) in the determinant equation yields 
(\ref{omega-definition}).	
\end{proof}

The next result gives the explicit expression for $\omega(z)$ 
for $\zeta$ in the band gap $(\zeta_-,\zeta_+)$ by using (\ref{z}) and (\ref{omega-definition}).

\begin{proposition}
	\label{prop-4}
	Let $z$ be given by the parameterization (\ref{z}). Then we have 
\begin{equation}
\label{omega-expression}
\omega = -2 k (1+k)^2 \frac{1 - k \sn^2(\alpha)}{[1+k\sn^2(\alpha)]^3} \sn(\alpha)\cn(\alpha) \dn(\alpha).
\end{equation}
\end{proposition}

\begin{proof}
Compatibility of (\ref{ZS-norm}) and (\ref{ZS-time}) implies that the expression for $\omega$ can be computed from the first algebraic equation of system (\ref{ZS-time}) at a single value of $x$, e.g. at $x = 0$. By doing so, we obtain 
\begin{align*}
\omega + i \zeta (1+k^2) - 4 i \zeta^3 &= -2i\zeta k \frac{q(0;z)}{p(0;z)} \\
&= 2i \zeta k \frac{\Theta^2\left(\alpha - \frac{iK'}{2}\right)}{
	H^2\left(\alpha - \frac{i K'}{2}\right)} \\
&= 2i \zeta k \frac{(1+k) \sn(\alpha) + i \cn(\alpha) \dn(\alpha)}{(1+k) \sn(\alpha) - i \cn(\alpha) \dn(\alpha)},
\end{align*}
where we have used (\ref{sn-theta-quotient}),  (\ref{p}), (\ref{z}), and (\ref{appendix-quarter-period}). Expressing now $\zeta$ in terms of $\alpha$ by using (\ref{char-prop}) in the band gap $(\zeta_-,\zeta_+)$ 
and using  (\ref{appendix-extra}), we obtain 
\begin{align*}
\omega &= i \zeta \left[ 4 \zeta^2 - 1 - k^2 + 2k \frac{[(1+k) \sn(\alpha) + i \cn(\alpha) \dn(\alpha)]^2}{[1 + k \sn^2(\alpha)]^2} \right] \\
&= - 4 \zeta k (1+k) \frac{\sn(\alpha) \cn(\alpha) \dn(\alpha)}{[1 + k \sn^2(\alpha)]^2}, 
\end{align*}
which yields (\ref{omega-expression}). 
\end{proof}

\begin{remark}
	It follows from (\ref{omega-expression}) that the exponent $e^{\omega(z) t}$ in (\ref{time-evol}) is purely real for $\alpha \in (0,K)$. The values of $s$ in (\ref{s-expression}) and $\omega$ in (\ref{omega-expression}) are strictly negative for all $\alpha \in (0,K)$.
\end{remark}

\begin{remark}
The expression (\ref{omega-expression}) can be obtained by taking the negative square root from the expression (\ref{omega-definition}) after $\zeta$ is expressed by (\ref{char}).
\end{remark}

\subsection{Quarter-period translation of the Jacobi's theta functions}

The relevance of the quarter-period translations of Jacobi's theta functions follows from the representations (\ref{p}) with (\ref{z}):
\begin{equation*}
H(x-iz) = H\left(x+\alpha-\frac{i K'}{2}\right), \quad 
\Theta(x-iz) = \Theta\left(x+\alpha-\frac{i K'}{2}\right).
\end{equation*}
The following proposition specifies some useful quarter-period translation formulas of the Jacobi's theta functions, which are novel 
to the best of our knowledge.

\begin{proposition}
	\label{prop-5}
We have for every $x \in \mathbb{R}$:	
\begin{equation}
	\label{squared-rel-2}
\frac{H^2\left(x + \frac{i K'}{2}\right) \Theta^2\left(\frac{i K'}{2}\right)}{\Theta^2(x) \Theta^2(0)} = \frac{i \sqrt{k}}{2(1+k)} e^{\frac{\pi K'}{4K}}e^{-\frac{i \pi x}{2K}} \left[ (1+k) \sn(x) + i \cn(x) \dn(x) \right]
	\end{equation}
	and
\begin{equation}
\label{squared-rel-1}
\frac{\Theta^2\left(x + \frac{iK'}{2}\right) \Theta^2\left(\frac{i K'}{2}\right)}{\Theta^2(x) \Theta^2(0)} = \frac{i \sqrt{k}}{2(1+k)} e^{\frac{\pi K'}{4K}}e^{-\frac{i \pi x}{2K}} \left[ (1+k) \sn(x) - i \cn(x) \dn(x) \right].
\end{equation}
\end{proposition}

\begin{proof}
We start with the quadratic identities for Jacobi's theta functions \cite[(1.4.16) and (1.4.19)]{Lawden}:
\begin{equation}
\label{appendix-7}
\begin{array}{l}
H(x+y) H(x-y) \Theta^2(0) = H^2(x) \Theta^2(y) -\Theta^2(x) H^2(y), \\  
\Theta(x+y)  \Theta(x-y)  \Theta^2(0) =  \Theta^2(x)  \Theta^2(y) - H^2(x) H^2(y). 
\end{array}
\end{equation}
It follows from (\ref{sn-theta-quotient}) with $\sn(\frac{iK'}{2}) = \frac{i}{\sqrt{k}}$ that 
$$
H\left(\frac{i K'}{2}\right) = i  \Theta\left(\frac{i K'}{2}\right),
$$
where $\Theta\left(\frac{i K'}{2}\right)$ is real since $\Theta$ is even with real coefficients. Hence we obtain from (\ref{appendix-7}):
\begin{align*}
\left[ H^2(x) + \Theta^2(x)\right] \Theta^2\left(\frac{iK'}{2}\right) &= 
H\left(x+\frac{i K'}{2} \right) H\left(x-\frac{i K'}{2} \right) \Theta^2(0) \\
&= \Theta\left(x+\frac{i K'}{2} \right) \Theta\left(x-\frac{i K'}{2} \right) \Theta^2(0).
\end{align*}
By using the half-period translations of Jacobi's theta functions (\ref{appendix-8}), we obtain their squared quarter-period translations:
\begin{equation}
\label{squared-rel}
\left[ H^2\left(x + \frac{iK'}{2}\right) + \Theta^2\left(x + \frac{i K'}{2}\right) \right] \Theta^2\left(\frac{i K'}{2}\right) = i e^{\frac{\pi K'}{4K}} e^{-\frac{i \pi x}{2K}} H(x) \Theta(x) \Theta^2(0).
\end{equation}
Using (\ref{sn-theta-quotient}) and (\ref{appendix-quarter-period}) with 
\begin{align*}
1 + k \sn^2\left(x + \frac{i K'}{2}\right) &=  \frac{2 (1+k) \sn(x)}{(1+k) \sn(x) - i \cn(x) \dn(x)},
\end{align*}
we obtain (\ref{squared-rel-2}) and (\ref{squared-rel-1}) from (\ref{squared-rel}). 
\end{proof}

\begin{remark}
Setting $x = 0$ in (\ref{squared-rel-1}) yields the useful relation:
\begin{equation}
\label{squared-rel-3}
\frac{\Theta^4\left(\frac{i K'}{2}\right)}{\Theta^2(0)} = \frac{\sqrt{k}}{2(1+k)} e^{\frac{\pi K'}{4K}} .
\end{equation}
\end{remark}

\subsection{One-mode transformation of the snoidal potential}

Here we apply the one-fold Darboux transformation (\ref{DT}) with the  particular solution $\varphi = (p,q)^T$ of the linear system (\ref{spectral}) and (\ref{time_evol}) given by (\ref{p}) and (\ref{time-evol}). The following proposition contains an important identity for the relevant computations of the two-mode transformation. 

\begin{proposition}
	\label{prop-6}
For every $x, \alpha \in \mathbb{R}$, we have 
\begin{align}
\notag
& \sn(x) \sn(\alpha) \cn(x+\alpha) \dn(x+\alpha) + \sn(x) \sn(x+\alpha) \cn(\alpha) \dn(\alpha) \\
& \qquad + \sn^2(\alpha) - \sn^2(x+\alpha) = 0.
\label{A-expression-zero}
\end{align}	
\end{proposition}

\begin{proof}
We expand the left-hand side of (\ref{A-expression-zero}) with the addition formulas (\ref{appendix-5}):
\begin{align*}
& \frac{\sn(x) \sn(\alpha)   [\cn(x) \cn(\alpha) - \sn(x) \sn(\alpha) \dn(x) \dn(\alpha)][\dn(x) \dn(\alpha) - k^2 \sn(x) \sn(\alpha) \cn(x) \cn(\alpha)]}{[1 - k^2 \sn^2(x) \sn^2(\alpha)]^2} \\
& \qquad + \frac{\sn(x) \cn(\alpha) \dn(\alpha)  [\sn(x) \cn(\alpha) \dn(\alpha) + \sn(\alpha) \cn(x) \dn(x)]}{[1 - k^2 \sn^2(x) \sn^2(\alpha)]}\\
& \qquad  + \sn^2(\alpha)
- \frac{[\sn(x) \cn(\alpha) \dn(\alpha) + \sn(\alpha) \cn(x) \dn(x)]^2}{[1 - k^2 \sn^2(x) \sn^2(\alpha)]^2}.
\end{align*}
Expanding the numerators of the two quotients with the squared denominators 
yields a simplification
\begin{align*}
& -\frac{\sn(x) \sn(\alpha) \cn(x) \cn(\alpha) \dn(x) \dn(\alpha) [1 - k^2 \sn^2(x) \sn^2(\alpha)] + {\rm Rem}}{[1 - k^2 \sn^2(x) \sn^2(\alpha)]^2},
\end{align*}
where
\begin{align*}
{\rm Rem} &:= \sn^2(x) \sn^2(\alpha) \dn^2(x) \dn^2(\alpha)  + k^2 \sn^2(x) \sn^2(\alpha) \cn^2(x) \cn^2(\alpha)\\
& \quad + \sn^2(x) \cn^2(\alpha) \dn^2(\alpha) + \sn^2(\alpha) \cn^2(x) \dn^2(x)
\end{align*}
is also divisible by $[1  - k^2 \sn^2(x) \sn^2(\alpha)]$ due to the explicit factorization:
\begin{align*}
{\rm Rem} &= [\cn^2(x) \sn^2(\alpha) + \sn^2(x) \dn^2(\alpha)] [1 - k^2 \sn^2(x) \sn^2(\alpha)].
\end{align*}
This allows us to rewrite the left-hand side of (\ref{A-expression-zero}) in the simplified form:
\begin{align*}
\frac{\sn^2(x) \cn^2(\alpha) \dn^2(\alpha) - \cn^2(x) \sn^2(\alpha) - \sn^2(x) \dn^2(\alpha)}{[1 - k^2 \sn^2(x) \sn^2(\alpha)]} + \sn^2(\alpha).
\end{align*}
The numerator of the first quotient is divisible by $[1  - k^2 \sn^2(x) \sn^2(\alpha)]$:
\begin{align*}
& \sn^2(x) \cn^2(\alpha) \dn^2(\alpha) - \cn^2(x) \sn^2(\alpha) - \sn^2(x) \dn^2(\alpha) \\
\quad &= -\sn^2(\alpha) [\cn^2(x) + \sn^2(x) \dn^2(\alpha)] \\
\quad &= -\sn^2(\alpha) [1 - k^2 \sn^2(x) \sn^2(\alpha)],
\end{align*}
which completes the proof of (\ref{A-expression-zero}) with $-\sn^2(\alpha) + \sn^2(\alpha) = 0$.	
\end{proof}

With the help of Proposition \ref{prop-6}, we prove that the transformation 
(\ref{DT}) with the one-mode solution (\ref{p}) and (\ref{time-evol}) recovers the same snoidal potential $\phi_0(x) = k \sn(x)$.

\begin{proposition}
	\label{prop-7}
Let $z$ be given by (\ref{z}), $\zeta$ be given by (\ref{char-prop}),  and $\varphi = (p,q)^T$ be given by (\ref{p}) and (\ref{time-evol}) with $s$ and $\omega$ in (\ref{s-expression}) and (\ref{omega-expression}). The transformation formula (\ref{DT}) with $u(x) = k \sn(x)$ returns 
\begin{equation}
\label{single-mode}
\hat{u}(x) = -\frac{1}{\sn(x + 2 \alpha)} \quad \Rightarrow \quad 
-\hat{u}(x + i K') = k \sn(x+2\alpha),
\end{equation}
where $x$ stands for $x + c_0 t$ with $c_0 = 1 + k^2$.
\end{proposition}

\begin{proof}
By using (\ref{sn-theta-quotient}) and (\ref{squared-rel}), we obtain 
from (\ref{p}) and (\ref{time-evol}):
	\begin{align*}
	pq  &= e^{2 s x + 2\omega t} e^{-\frac{i \pi x}{2K}} 
	\frac{H\left( x+\alpha-\frac{i K'}{2} \right) 
		\Theta\left( x+\alpha-\frac{i K'}{2} \right)}{\Theta^2(x) 
		H\left( -\alpha + \frac{i K'}{2} \right) 
		\Theta\left( -\alpha + \frac{i K'}{2} \right) }\\
	&= -e^{2 s x + 2\omega t} \frac{H^2(x + \alpha) + \Theta^2(x+\alpha)}{\Theta^2(x) [H^2(\alpha) + \Theta^2(\alpha)]} \\
	&= -e^{2 s x + 2\omega t} \frac{\Theta^2(x+\alpha)}{
\Theta^2(x) \Theta^2(\alpha)} \frac{1 + k \sn^2(x+\alpha)}{1+k \sn^2(\alpha)}.
	\end{align*}
Since the complex exponential $e^{-\frac{i \pi x}{2K}}$ cancels out, we confirm that $pq$ is real since $e^{2s(z) x}$ and $e^{2 \omega(z) t}$ are real.
	
Similarly, by using (\ref{squared-rel-2}) and (\ref{squared-rel-1}), we obtain 
	\begin{align*}
	p^2 &= e^{2 s x + 2\omega t}  e^{-\frac{i \pi x}{2K}} \frac{H^2\left( x+\alpha -\frac{i K'}{2}\right)}{\Theta^2(x) \Theta^2\left(- \alpha + \frac{i K'}{2}\right)} \\
	&=  e^{2 s x + 2\omega t}
	\frac{(1+k) \sn(x+\alpha) -i \cn(x+\alpha) \dn(x+\alpha)}{(1+k) \sn(\alpha) + i \cn(\alpha) \dn(\alpha)}  \frac{\Theta^2(x+\alpha)}{\Theta^2(x) \Theta^2(\alpha)}
	\end{align*}
	and
	\begin{align*}
	q^2 &= e^{2 s x + 2\omega t}  e^{-\frac{i \pi x}{2K}} \frac{\Theta^2\left( x+\alpha -\frac{i K'}{2}\right)}{\Theta^2(x) H^2\left(- \alpha + \frac{i K'}{2}\right)} \\
	&= e^{2 s x + 2\omega t} 
	\frac{(1+k) \sn(x+\alpha) +i \cn(x+\alpha) \dn(x+\alpha)}{(1+k) \sn(\alpha) - i \cn(\alpha) \dn(\alpha)} \frac{\Theta^2(x+\alpha)}{\Theta^2(x) \Theta^2(\alpha)}.
	\end{align*}
Again, since the complex exponential $e^{-\frac{i \pi x}{2K}}$ cancels out, we confirm that $q^2$ is the complex conjugate of $p^2$.
	
Substituting explicit expressions for $pq$, $p^2$, and $q^2$ into (\ref{DT}) and using (\ref{char-prop}) with (\ref{appendix-5}), we obtain the new solution to the mKdV equation (\ref{1})
	\begin{align*}
	\hat{u}(x) &= k \sn(x) - \frac{4i\zeta pq}{p^2-q^2} \\
	&= k \sn(x) - \frac{2 \zeta}{(1+k)} \frac{[1 + k \sn^2(\alpha)][1 + k \sn^2(x+\alpha)]}{\sn(\alpha) \cn(x+\alpha) \dn(x+\alpha) + \sn(x+\alpha) \cn(\alpha) \dn(\alpha)} \\
	&= k \sn(x) - \frac{[1 - k \sn^2(\alpha)][1 + k \sn^2(x+\alpha)]}{
		\sn(x+2\alpha) [1 - k^2 \sn^2(\alpha)\sn^2(x + \alpha)]} \\
	&= -\frac{1}{\sn(x+2\alpha)} + k \left[ \sn(x) + \frac{\sn^2(\alpha) -  \sn^2(x+\alpha)}{\sn(\alpha) \cn(x+\alpha) \dn(x+\alpha) + \sn(x+\alpha) \cn(\alpha) \dn(\alpha)} \right] \\
	&= -\frac{1}{\sn(x+2\alpha)},
	\end{align*}
	where the expression in the brackets is identically equal to zero by (\ref{A-expression-zero}). Translating the new solution by $i K'$ with (\ref{appendix-2}), we obtain (\ref{single-mode}). 	
\end{proof}

\begin{remark}
The new solution 
$$
\tilde{u}(x,t) := -\hat{u}(x+iK'+c_0t) = \phi_0(x+c_0t + 2 \alpha)
$$ 
in (\ref{single-mode}) coincides with the same solution $u(x,t) = \phi_0(x+c_0t)$ after the translation along the real axis to the left by the phase shift $2 \alpha$.
\end{remark}

\begin{remark}
	If we use the second linearly independent solution of the linear system (\ref{ZS-norm}) given by (\ref{p}) and (\ref{time-evol}) with $z'$ instead of $z$ in the transformation formula (\ref{DT}), then $\hat{u}$ is given by (\ref{single-mode}) with $2 \alpha$ being replaced by $-2\alpha$. It is translated  along the real axis to the right by the same phase shift $2\alpha$.
\end{remark}

\subsection{Two-mode transformation of the snoidal potential}

In order to obtain a nontrivial solution $\hat{u}$ describing a soliton moving on the snoidal background, we take a linear superposition of two linearly independent solutions given by (\ref{p}) and (\ref{time-evol}) with $z$ and $z'$. The two solutions are only different by the sign of $\alpha$ in the parameterization (\ref{z}) which leaves the solution written in the same form but with the opposite signs of $s$ and $\omega$ in (\ref{s-expression}) and (\ref{omega-expression}). Hence we write:
\begin{equation}
\label{p-expression}
p = c_+ e^{sx + \omega t} e^{- \frac{i \pi x}{4K}} 	\frac{H\left( x+\alpha-\frac{i K'}{2} \right)}{\Theta(x) \Theta\left( -\alpha + \frac{i K'}{2} \right) } + c_- 
e^{-sx - \omega t} e^{- \frac{i \pi x}{4K}}
\frac{H\left( x - \alpha-\frac{i K'}{2} \right)}{\Theta(x) \Theta\left( \alpha + \frac{i K'}{2} \right) }
\end{equation}
and
\begin{equation}
\label{q-expression}
q = c_+ e^{sx + \omega t}  e^{- \frac{i \pi x}{4K}}
\frac{\Theta\left( x+\alpha-\frac{i K'}{2} \right)}{\Theta(x) H\left( -\alpha + \frac{i K'}{2} \right) } + c_- 
e^{-sx - \omega t} e^{- \frac{i \pi x}{4K}}
\frac{\Theta\left( x - \alpha - \frac{i K'}{2} \right)}{\Theta(x) H\left( \alpha + \frac{i K'}{2} \right) },
\end{equation}
where $c_+$ and $c_-$ are arbitrary constants. We choose 
\begin{equation}
\label{choice-c}
c_+ = \mathfrak{c} e^{s x_0}, \quad c_- = \mathfrak{c} e^{-sx_0},
\end{equation}
and define $\eta := -s (x+x_0) -\omega  t$. The following proposition contains an important identity in the derivation of the solution form in Theorem \ref{theorem-1}. 

\begin{proposition}
	\label{prop-8}
	For every $x, \alpha \in \mathbb{R}$, we have 
\begin{align}
& \quad - i \frac{\sn(x + \alpha - \frac{iK'}{2}) \sn(\alpha - \frac{iK'}{2}) - \sn(x - \alpha - \frac{iK'}{2}) \sn(\alpha + \frac{iK'}{2})}{\sn(\alpha + \frac{iK'}{2}) \sn(\alpha - \frac{iK'}{2})}  \notag \\
& \quad \times  \left[ (1+k) \sn(x) [1 - k \sn^2(\alpha)] + i \cn(x) \dn(x) [1 + k \sn^2(\alpha)] \right] \left[ 1 + k \sn^2(\alpha) \right] \notag \\
& = 2 \left[ (1+k)^2\sn^2(\alpha)(1-k\sn^2(x)) - \cn^2(\alpha) \dn^2(\alpha) 
(1+k\sn^2(x)) \right], \label{to-prove}
\end{align}
\end{proposition}

\begin{proof}
By using (\ref{appendix-extra}) and (\ref{appendix-5}), we obtain 
	\begin{align*}
	& \sn\left(\alpha + \frac{iK'}{2}\right) \sn\left(\alpha - \frac{iK'}{2} \right) \\
	&= \frac{[(1+k) \sn(\alpha) + i \cn(\alpha) \dn(\alpha)][(1+k)\sn(\alpha) - i \cn(\alpha) \dn(\alpha)]}{k[1 + k \sn^2(\alpha)]^2} \\
	&= \frac{(1+k)^2 \sn^2(\alpha) + \cn^2(\alpha) \dn^2(\alpha)}{k[1 + k \sn^2(\alpha)]^2} = \frac{1}{k}
	\end{align*}
	and
	\begin{align*}
	& \quad \frac{\sn(x + \alpha - \frac{iK'}{2}) \sn(\alpha - \frac{iK'}{2}) - \sn(x - \alpha - \frac{iK'}{2}) \sn(\alpha + \frac{iK'}{2})}{\sn(\alpha + \frac{iK'}{2}) \sn(\alpha - \frac{iK'}{2})}  \\
	&= \frac{[(1+k) \sn(x + \alpha) - i \cn(x+\alpha) \dn(x + \alpha)][(1+k) \sn(\alpha) - i \cn(\alpha) \dn(\alpha)]}{[1 + k \sn^2(x+\alpha)] [1 + k \sn^2(\alpha)]} \\
	& \quad - \frac{[(1+k) \sn(x - \alpha) - i \cn(x-\alpha) \dn(x - \alpha)][(1+k) \sn(\alpha) + i \cn(\alpha) \dn(\alpha)]}{[1 + k \sn^2(x-\alpha)] [1 + k \sn^2(\alpha)]}
	\end{align*}
We will use Landen transformation formulas:
\begin{equation}
\begin{array}{l}
\displaystyle
\sn\left((1+k)x;\frac{2\sqrt{k}}{1+k}\right) = \frac{(1+k) \sn(x;k)}{1 + k \sn^2(x;k)}, \\  
\displaystyle
\cn\left((1+k)x;\frac{2\sqrt{k}}{1+k}\right) = \frac{\cn(x;k) \dn(x;k)}{1 + k \sn^2(x;k)}, \\
\displaystyle
\dn\left((1+k)x;\frac{2\sqrt{k}}{1+k}\right) = \frac{1 - k\sn^2(x;k)}{1 + k \sn^2(x;k)}, \end{array} 
		\label{appendix-9}
\end{equation}
where the new elliptic modulus $\kappa := 2 \sqrt{k}/ (1+k)$ and the old elliptic modulus $k$ are listed explicitly. Suppressing the second argument, the transformation formulas (\ref{appendix-9}) rewrite the previous expression in the new form:
	\begin{align*}
	& \quad \frac{\sn(x + \alpha - \frac{iK'}{2}) \sn(\alpha - \frac{iK'}{2}) - \sn(x - \alpha - \frac{iK'}{2}) \sn(\alpha + \frac{iK'}{2})}{\sn(\alpha + \frac{iK'}{2}) \sn(\alpha - \frac{iK'}{2})}  \\
	&= [ \sn(1+k)(x+\alpha) - i \cn(1+k)(x+\alpha)] \; [\sn(1+k)\alpha - i \cn(1+k)\alpha] \\
	& \quad - [ \sn(1+k)(x-\alpha) - i \cn(1+k)(x-\alpha)] \; [\sn(1+k)\alpha + i \cn(1+k)\alpha].
	\end{align*}
By using  (\ref{appendix-9}) and (\ref{appendix-5}) backwards, we obtain 
	\begin{align*}
	& \quad \frac{\sn(x + \alpha - \frac{iK'}{2}) \sn(\alpha - \frac{iK'}{2}) - \sn(x - \alpha - \frac{iK'}{2}) \sn(\alpha + \frac{iK'}{2})}{\sn(\alpha + \frac{iK'}{2}) \sn(\alpha - \frac{iK'}{2})}  \\
	&= \frac{2 [\cn(1+k)x + i \sn(1+k)x \dn(1+k)\alpha] [\sn^2(1+k)\alpha \dn(1+k) x - \cn^2(1+k) \alpha]}{1 - \kappa^2 \sn^2(1+k)x \sn^2(1+k)\alpha} \\
	&= 2 [\cn x \dn x (1 + k \sn^2 \alpha) + i (1+k) \sn x (1 -  k \sn^2 \alpha)]  \\
    & \quad \times \frac{[(1+k)^2 \sn^2 \alpha (1-k \sn^2 x) - \cn^2 \alpha \dn^2 \alpha (1+k \sn^2 x)] }{[1 + k \sn^2 \alpha] [(1+k \sn^2 x)^2 (1 + k \sn^2 \alpha)^2 - 4 k (1+k)^2 \sn^2 x \sn^2 \alpha]}
	\end{align*}
Multiplying this formula by 
	$$
	-i [ (1+k) \sn x (1 - k \sn^2 \alpha) + i \cn x \dn x (1 + k \sn^2 \alpha) ] [1 + k \sn^2 \alpha]
	$$
	yields (\ref{to-prove}) if and only if the following identity holds:
	\begin{align*}
	& \quad
	(1+k)^2 \sn^2 x (1-k \sn^2 \alpha)^2 + \cn^2 x \dn^2 x (1 + k \sn^2 \alpha)^2 \\
	&= (1+k \sn^2 x)^2 (1 + k \sn^2 \alpha)^2 - 4 k (1+k)^2 \sn^2 x \sn^2 \alpha.
	\end{align*}
However, this identity is true in view of the fundamental relations 
(\ref{fund-rel}). Hence, the identity (\ref{to-prove}) has been proved.
\end{proof}

We can now provide the proof of Theorem \ref{theorem-1}. Since the parameter $\mathfrak{c}$ in (\ref{choice-c}) cancels in the quotient (\ref{DT}) and so is the common factor $\Theta^2(x)$ in the denominators of $p q$, $p^2$, and $q^2$,  we will not write these common factors and use the sign $\simeq$ for the equivalent expressions up to the division of these common factors. 

\begin{proof}[Proof of Theorem \ref{theorem-1}]
By using (\ref{p-expression}) and (\ref{choice-c}), we write
\begin{align*}
p^2 &\simeq e^{-2 \eta} e^{-\frac{i \pi x}{2K}} 	\frac{H^2\left( x+\alpha-\frac{i K'}{2} \right)}{\Theta^2\left( -\alpha + \frac{i K'}{2} \right) } + e^{2 \eta} e^{-\frac{i \pi x}{2K}}
\frac{H^2\left( x - \alpha-\frac{i K'}{2} \right)}{\Theta^2\left( \alpha + \frac{i K'}{2} \right) }\\
& \quad + 2 e^{-\frac{i \pi x}{2K}} \frac{H\left( x+\alpha-\frac{i K'}{2} \right) H\left( x - \alpha-\frac{i K'}{2} \right)}{\Theta\left( -\alpha + \frac{i K'}{2} \right) \Theta\left( \alpha + \frac{i K'}{2} \right)}. 
\end{align*}
By using (\ref{sn-theta-quotient}), (\ref{squared-rel-2}), (\ref{squared-rel-1}), (\ref{appendix-7}), and (\ref{squared-rel-3}), we obtain 
\begin{align*}
p^2 &\simeq e^{-2 \eta}
\frac{(1+k) \sn(x+\alpha) -i \cn(x+\alpha) \dn(x+\alpha)}{(1+k) \sn(\alpha) + i \cn(\alpha) \dn(\alpha)}  \frac{\Theta^2(x+\alpha)}{\Theta^2(\alpha)} \\
& \quad - e^{2 \eta}
\frac{(1+k) \sn(x-\alpha) -i \cn(x-\alpha) \dn(x-\alpha)}{(1+k) \sn(\alpha) - i \cn(\alpha) \dn(\alpha)}  \frac{\Theta^2(x-\alpha)}{\Theta^2(\alpha)} \\
& \quad - 2i  \frac{(1+k) \sn(x) [1 - k \sn^2(\alpha)] - i \cn(x) \dn(x) [1 + k \sn^2(\alpha)]}{1 + k \sn^2(\alpha)} \frac{\Theta^2(x)}{\Theta^2(0)}.
\end{align*}
Similarly, we obtain 
\begin{align*}
q^2 &\simeq e^{-2 \eta}
\frac{(1+k) \sn(x+\alpha) +i \cn(x+\alpha) \dn(x+\alpha)}{(1+k) \sn(\alpha) - i \cn(\alpha) \dn(\alpha)}  \frac{\Theta^2(x+\alpha)}{\Theta^2(\alpha)}\\
& \quad - e^{2 \eta}
\frac{(1+k) \sn(x-\alpha) +i \cn(x-\alpha) \dn(x-\alpha)}{(1+k) \sn(\alpha) + i \cn(\alpha) \dn(\alpha)}  \frac{\Theta^2(x-\alpha)}{\Theta^2(\alpha)} \\
& \quad + 2i \frac{(1+k) \sn(x) [1 - k \sn^2(\alpha)] + i \cn(x) \dn(x) [1 + k \sn^2(\alpha)]}{1 + k \sn^2(\alpha)} \frac{\Theta^2(x)}{\Theta^2(0)}.
\end{align*}
This yields a compact expression for the real-valued quantity in the transformation (\ref{DT}):
\begin{align*}
\frac{p^2 - q^2}{-2i(1+k)} &\simeq e^{-2 \eta}
\frac{\sn(x+\alpha) \cn(\alpha) \dn(\alpha) + \sn(\alpha) \cn(x+\alpha) \dn(x+\alpha)}{[1 + k \sn^2(\alpha)]^2}  \frac{\Theta^2(x+\alpha)}{\Theta^2(\alpha)}\\
& \quad + e^{2 \eta}
\frac{\sn(x-\alpha) \cn(\alpha) \dn(\alpha) - \sn(\alpha) \cn(x - \alpha) \dn(x-\alpha)}{[1 + k \sn^2(\alpha)]^2}  \frac{\Theta^2(x-\alpha)}{\Theta^2(\alpha)} \\
& \quad + 2 \sn(x) \frac{1 - k \sn^2(\alpha)}{1 + k \sn^2(\alpha)} \frac{\Theta^2(x)}{\Theta^2(0)}.
\end{align*}
To simplify the expression for $p q$, we write
\begin{align*}
p q &\simeq e^{-2 \eta} e^{-\frac{i \pi x}{2K}}
\frac{H\left( x+\alpha-\frac{i K'}{2} \right) \Theta\left( x+\alpha-\frac{i K'}{2} \right)}{H\left( -\alpha + \frac{i K'}{2} \right) \Theta\left( -\alpha + \frac{i K'}{2} \right) } \\
& \quad + e^{2 \eta} e^{-\frac{i \pi x}{2K}} \frac{H\left( x - \alpha-\frac{i K'}{2} \right) \Theta\left( x - \alpha-\frac{i K'}{2} \right)}{H\left( \alpha + \frac{i K'}{2} \right) \Theta\left( \alpha + \frac{i K'}{2} \right) }  \\
& \quad + e^{-\frac{i \pi x}{2K}} \frac{H\left( x+\alpha-\frac{i K'}{2} \right) \Theta\left( x - \alpha-\frac{i K'}{2} \right)}{H\left( \alpha + \frac{i K'}{2} \right) \Theta\left( -\alpha + \frac{i K'}{2} \right) }  \\
& \quad + e^{-\frac{i \pi x}{2K}} \frac{H\left( x - \alpha-\frac{i K'}{2} \right) \Theta\left( x+\alpha-\frac{i K'}{2} \right)}{H\left( -\alpha + \frac{i K'}{2} \right) \Theta\left( \alpha + \frac{i K'}{2} \right) } .
\end{align*}
Using (\ref{sn-theta-quotient}),  (\ref{squared-rel-2}), (\ref{squared-rel-1}), and (\ref{squared-rel-3}), we obtain
\begin{align*}
p q &\simeq -e^{-2 \eta}  \frac{1 + k \sn^2(x+\alpha)}{1+k \sn^2(\alpha)} 
\frac{\Theta^2\left(x+\alpha\right)}{\Theta^2\left( \alpha\right)} \\
& \quad -e^{2 \eta}  \frac{1 + k \sn^2(x-\alpha)}{1+k \sn^2(\alpha)} 
\frac{\Theta^2\left(x-\alpha\right)}{\Theta^2\left( \alpha\right)}\\
& \quad - i \left[ \frac{\sn(x + \alpha - \frac{iK'}{2})}{\sn(\alpha + \frac{iK'}{2})} + \frac{\sn(x - \alpha - \frac{iK'}{2})}{\sn(-\alpha + \frac{iK'}{2})} \right] \\
& \quad  \times  \frac{(1+k) \sn(x) [1 - k \sn^2(\alpha)] + i \cn(x) \dn(x) [1 + k \sn^2(\alpha)]}{1 + k \sn^2(\alpha)} \frac{\Theta^2\left(x\right)}{\Theta^2(0)}.
\end{align*}
By using (\ref{to-prove}), we rewrite $p q$ as follows:
\begin{align*}
p q &\simeq -e^{-2 \eta}  \frac{1 + k \sn^2(x+\alpha)}{1+k \sn^2(\alpha)} 
\frac{\Theta^2(x+\alpha)}{\Theta^2(\alpha)} \\
& \quad -e^{2 \eta}  \frac{1 + k \sn^2(x-\alpha)}{1+k \sn^2(\alpha)} 
\frac{\Theta^2(x-\alpha)}{\Theta^2(\alpha)}\\
& \quad + 2 \frac{(1+k)^2\sn^2(\alpha)(1-k\sn^2(x)) - \cn^2(\alpha) \dn^2(\alpha)(1+k\sn^2(x))}{(1 + k \sn^2(\alpha))^2} \frac{\Theta^2(x)}{\Theta^2(0)}.
\end{align*}
Next we substitute the simplified expressions for $p q$ and 
$(p^2-q^2)/(-2i(1+k))$ into (\ref{DT}) and bring it to the common denominator. The resulting formula is greatly simplified with the use of (\ref{appendix-extra}) and (\ref{A-expression-zero}) to the form 
$\hat{u} = -N/D$, where 
\begin{align*}
N(x) &= e^{-2 \eta}  [1 - k^2 \sn^2(\alpha) \sn^2(x+\alpha)] 
\frac{\Theta^2(x+\alpha)}{\Theta^2(\alpha)} \\
& \quad + e^{2 \eta}  [1 - k^2 \sn^2(\alpha) \sn^2(x-\alpha)] 
\frac{\Theta^2(x-\alpha)}{\Theta^2(\alpha)} \\
& \quad + 2 \frac{1 - k \sn^2(\alpha)}{1 + k \sn^2(\alpha)}
[	\cn^2(\alpha) \dn^2(\alpha) - (1+k)^2\sn^2(\alpha)]  \frac{\Theta^2(x)}{\Theta^2(0)}
\end{align*}
and 
\begin{align*}
D(x) &= e^{-2 \eta}  [1 - k^2 \sn^2(\alpha) \sn^2(x+\alpha)] 
\frac{\Theta^2(x+\alpha)}{\Theta^2(\alpha)} \sn(x+2\alpha) \\
& \quad + e^{2 \eta}  [1 - k^2 \sn^2(\alpha) \sn^2(x-\alpha)] 
\frac{\Theta^2(x-\alpha)}{\Theta^2(\alpha)} \sn(x-2\alpha) \\
& \quad + 2 (1 - k^2 \sn^4(\alpha)) \frac{\Theta^2(x)}{\Theta^2(0)} \sn(x).
\end{align*}
Multiplying both $N$ and $D$ by $\Theta^4(\alpha)$, we rewrite $\hat{u}$ in the form:
\begin{equation*}
\hat{u}(x) = - 
\frac{G(x+\alpha,\alpha) e^{-2\eta}  + G(x-\alpha,\alpha) e^{2\eta} + 2 \beta G(x,0)}{G(x+\alpha,\alpha) e^{-2\eta} \sn(x+2\alpha) + 
	G(x-\alpha,\alpha) e^{2\eta} \sn(x-2\alpha) + 2 \gamma G(x,0) \sn(x)},
\end{equation*}
where $G(x,\alpha)$ is given by 
\begin{align}
\label{G-function}
G(x,\alpha) := \left[ 1 - k^2 \sn^2(x) \sn^2(\alpha) \right] \Theta^2(x)\Theta^2(\alpha)
\end{align}
and the $x$-independent parameters $\beta$ and $\gamma$ are given by 
\begin{align*}
\beta &:= \frac{1-k \sn^2(\alpha)}{1+k \sn^2(\alpha)}  [\cn^2(\alpha) \dn^2(\alpha) - (1+k)^2 \sn^2(\alpha)] \frac{\Theta^4(\alpha)}{\Theta^4(0)}, \\
\gamma &:= \left[ 1 - k^2 \sn^4(\alpha) \right] \frac{\Theta^4(\alpha)}{\Theta^4(0)}.
\end{align*}
The expresion for $\hat{u}(x)$ reduces to the one-mode solutions (\ref{single-mode}) as $\eta \to \pm \infty$, which is non-singular after the half-period translation along the imaginary axis: 
$$
\hat{u}(x) \to \hat{u}(x+iK').
$$
To perform the same half-period translation of $\hat{u}(x)$ for the two-mode solution, we use the translation formulas (\ref{appendix-8}) in $G(x,\alpha)$ rewritten as 
$$
G(x,\alpha) = \Theta^2(x)\Theta^2(\alpha) - H^2(x)H^2(\alpha).
$$ 
After canceling the numerical and $x$-dependent factors in the quotient for $\hat{u}(x + i K')$ and the complex phase in $\eta$ by a suitable choice of $x_0$, 
the expression for $\hat{u}(x+iK')$ is transformed to the same form as for $\hat{u}(x)$ with  $G(x \pm \alpha+iK',\alpha)$ being replaced by $\hat{G}(x \pm \alpha,\alpha)$. Simplifying $\hat{G}(x + \alpha,\alpha)$ yields
\begin{align*}
\hat{G}(x+\alpha,\alpha) &= H^2(x+\alpha)\Theta^2(\alpha) - \Theta^2(x+\alpha) H^2(\alpha)\\
&= k \Theta^2(x+\alpha)\Theta^2(\alpha) [\sn^2(x+\alpha)-\sn^2(\alpha)] \\
&= k \Theta^2(x+\alpha) \Theta^2(\alpha) \sn(x) [\sn(\alpha) \cn(x+\alpha) \dn(x+\alpha) + \sn(x+\alpha) \cn(\alpha) \dn(\alpha)] \\
&= k \Theta^2(x+\alpha) \Theta^2(\alpha) \sn(x) \sn(x+2\alpha) [1 - k^2 \sn^2(\alpha) \sn^2(x+\alpha)] \\
&= k \sn(x) \sn(x+2\alpha) G(x+\alpha,\alpha),
\end{align*}
where we have used again (\ref{A-expression-zero}) and where $G(x,\alpha)$ is given by (\ref{G-function}). Substituting $\hat{G}(x\pm\alpha,\alpha)$, and $\hat{G}(x,0)$ into $\hat{u}(x + i K')$ and canceling one power of $\sn(x)$ yields the formula 
\begin{align*}
& \hat{u}(x+iK') \\
& = -k 
\frac{G(x+\alpha,\alpha) e^{-2\eta} \sn(x+2\alpha) + G(x-\alpha,\alpha) e^{2\eta} \sn(x-2\alpha) + 2 \beta G(x,0) \sn(x)}{G(x+\alpha,\alpha) e^{-2\eta} + G(x-\alpha,\alpha) e^{2\eta}  + 2 \gamma G(x,0)}.
\end{align*}
It follows from (\ref{appendix-7}) that 
$$
G(x,\alpha) = \Theta^2(0) \Theta(x+\alpha) \Theta(x-\alpha).
$$
Canceling $\Theta^2(0) \Theta(x)$ yields 
\begin{align*}
\hat{u}(x+iK') = -k 
\frac{\Theta(x+2\alpha) \sn(x+2\alpha) e^{-2\eta} + 
	\Theta(x-2\alpha) \sn(x-2\alpha) e^{2\eta} + 2 \beta \Theta(x) \sn(x)}{\Theta(x+2\alpha) e^{-2\eta} + \Theta(x-2\alpha) e^{2\eta}  + 2 \gamma \Theta(x)}.
\end{align*}
Similarly, we simplify the constants $\beta$ and $\gamma$:
$$
\gamma = \frac{G(\alpha,\alpha)}{\Theta^4(0)} = \frac{\Theta(2\alpha)}{\Theta(0)}
$$
and
\begin{align*}
\beta &= \frac{1-k \sn^2(\alpha)}{1+k \sn^2(\alpha)}  [(1+k \sn^2(\alpha))^2 - 2 (1+k)^2 \sn^2(\alpha)] \frac{\Theta^4(\alpha)}{\Theta^4(0)} \\
&= \left[ 1 - \frac{2 (1+k)^2 \sn^2(\alpha)}{(1 + k \sn^2(\alpha))^2} \right] \frac{\Theta(2\alpha)}{\Theta(0)},
\end{align*}
where we have used (\ref{appendix-extra}). These expressions coincide with (\ref{beta}). Finally, using (\ref{sn-theta-quotient}) and canceling the negative sign in $\hat{u}(x+iK')$ yields the two-mode solution in the final form (\ref{new-solution}), where we recall that $x$ is replaced by $\xi = x + c_0 t$ and $\eta = -s (x + c_0 t +x_0) - \omega t$ is expressed explicitly from (\ref{s-expression}) and (\ref{omega-expression}).
\end{proof}

We finish by giving the proof of Corollary \ref{theorem-2}. 

\begin{proof}[Proof of Corollary \ref{theorem-2}]
We take the limit $k \to 1$ with the help of the limiting relations:
\begin{equation*}
		k = 1: \quad \sn(x) = \tanh(x), \quad \cn(x) = {\rm sech}(x), \quad 
		\dn(x) = {\rm sech}(x)
\end{equation*}
and 
\begin{align*}
k = 1: \quad \beta = [1 - \sinh^2(2\alpha)] {\rm sech}(2\alpha), \quad \gamma = \cosh(2 \alpha).
\end{align*}
The asymptotic behavior $\Theta(x) \simeq \cosh(x)$ with some $k$-dependent numerical factor has been clarified in \cite[Eq. (36)]{HMP}. It follows from (\ref{sn-theta-quotient}) that $H(x) \simeq \sinh(x)$ with the same $k$-dependent factor. Taking the limit $k \to 1$ in (\ref{new-solution}) and cancelling the $k$-dependent numerical factor in the quotient yields the expression (\ref{two-solitons}). We also obtain:
$$
k = 1 : \quad c_0 = 2, \quad 
\kappa = \tanh(\alpha) + \tanh(\alpha) {\rm sech}(2\alpha) = \tanh(2\alpha), 
$$
and $c = 2 + 4 \; {\rm sech}^2(2\alpha)$.
\end{proof}

\appendix
\section{Relations between Jacobi's elliptic functions}
\label{sec-app}

Jacobi's elliptic functions satisfy the translation properties \cite[(2.2.17)--(2.2.19)]{Lawden}:
\begin{equation}
\label{appendix-4}
\sn(x+K) = \frac{\cn(x)}{\dn(x)}, \quad 
\cn(x+K) = \frac{-i k' \sn(x)}{\dn(x)}, \quad 
\dn(x+K) = \frac{k'}{\dn(x)}
\end{equation} 
and
\begin{equation}
\label{appendix-2}
\sn(x+iK') = \frac{1}{k \sn(x)}, \quad 
\cn(x+iK') = \frac{-i \dn(x)}{k \sn(x)}, \quad 
\dn(x+iK') = \frac{-i \cn(x)}{\sn(x)},
\end{equation}
as well as the reflection formulas \cite[(2.6.12)]{Lawden} with $k' = \sqrt{1-k^2}$:
\begin{equation}
\label{appendix-1}
\sn(ix;k) = \frac{i \sn(x;k')}{\cn(x;k')}, \quad 
\cn(ix;k) = \frac{1}{\cn(x;k')}, \quad 
\dn(ix;k) = \frac{\dn(x;k')}{\cn(x;k')}.
\end{equation}
Similarly, Jacobi's theta functions satisfy the translation properties 
\cite[(1.3.6) and (1.3.9)]{Lawden}:
\begin{equation}
\label{appendix-8}
H(x + i K') = i e^{\frac{\pi K'}{4K}} e^{-\frac{i \pi x}{2K}} \Theta(x), \quad
\Theta(x + i K') = i e^{\frac{\pi K'}{4K}} e^{-\frac{i \pi x}{2K}} H(x),
\end{equation}
The addition formulas for elliptic functions \cite[(2.4.1)--(2.4.3)]{Lawden} are given by 
\begin{equation}
\begin{array}{l}
\sn(u\pm v) = \displaystyle \frac{\sn(u)\cn(v)\dn(v) \pm \sn(v)\cn(u)\dn(u)}{1-k^2\sn^2(u)\sn^2(v)}, \\
\cn(u\pm v) = \displaystyle \frac{\cn(u)\cn(v) \mp \sn(u)\sn(v)\dn(u) \dn(v)}{1-k^2\sn^2(u)\sn^2(v)}, \\
\dn(u\pm v) = \displaystyle \frac{\dn(u)\dn(v)\mp k^2\sn(u)\sn(v)\cn(u)\cn(v)}{1-k^2\sn^2(u)\sn^2(v)},
\end{array} 
\label{appendix-5}
\end{equation}
from which we obtain the following translation formulas:
\begin{equation}
\begin{array}{l}
\sn\left(x + \frac{iK'}{2}\right) = \displaystyle \frac{1}{\sqrt{k}} \frac{(1+k) \sn(x) + i \cn(x) \dn(x)}{1 + k \sn^2(x)}, \\
\cn\left(x + \frac{iK'}{2}\right) = \displaystyle \frac{\sqrt{1+k}}{\sqrt{k}} \frac{\cn(x) - i \sn(x) \dn(x)}{1 + k \sn^2(x)}, \\
\dn\left(x + \frac{iK'}{2}\right) = \displaystyle \sqrt{1+k} \frac{\dn(x) - i k \sn(x) \cn(x)}{1 + k \sn^2(x)}.
\end{array} 
\label{appendix-quarter-period}
\end{equation}
Jacobi's zeta function satisfies the following addition formula \cite[(3.6.2)]{Lawden}:
\begin{equation} 
Z(u \pm v) = Z(u) \pm Z(v) \mp k^2\sn(u)\sn(v)\sn(u \pm v).
\label{appendix-6}
\end{equation} 

\vspace{0.2cm}

{\bf Conflict of Interest:} The authors declare that they have no conflict of interest.

\vspace{0.2cm}

{\bf Data availability statement:} The data is available upon request.


\begin{thebibliography}{99}
	
\bibitem{Ablowitz} M.J. Ablowitz, D.J. Kaup, A.C. Newell, and H. Segur, “The inverse scattering transform–Fourier analysis
for nonlinear problems”, Stud. Appl. Math. {\bf 53} (1974), 249--315

\bibitem {Ablowitz_2} M. J. Ablowitz, {\em Nonlinear Dispersive Waves: Asymptotic Analysis and Solitons} (Cambridge University Press, Cambridge, 2011)

\bibitem{Cole} M. J. Ablowitz, J. T. Cole, G. A. El, M. A. Hoefer, and
X. Luo, ``Soliton-mean field interaction in Korteweg-de Vries
dispersive hydrodynamics,'' Stud. Appl. Math. {\bf 151} (2023) 795--856

\bibitem{Dunke} G. Basar and G. V. Dunne, ``Twisted kink crystal in the chiral Gross-Neveu model", Phys. Rev. D {\bf 78} (2008) 065022 (21 pages)


\bibitem{Bertola} M. Bertola, R. Jenkins, and A. Tovbis, 
``Partial degeneration of finite gap solutions to the Korteweg--de
Vries equation: soliton gas and scattering on elliptic background", 
Nonlinearity {\bf 36} (2023) 3622--3660

\bibitem{Chen} J. Chen and D.E. Pelinovsky, ``Rogue periodic waves in the modified KdV equation", Nonlinearity {\bf 31} (2018) 1955--1980

\bibitem{CPnls} J. Chen and D.E. Pelinovsky, ``Rogue periodic waves in the focusing nonlinear Schr\"{o}dinger equation", Proc. R. Soc. Lond. A {\bf 474} (2018) 20170814 (18 pages).

\bibitem{Chen-JNLS} J. Chen and D.E. Pelinovsky, ``Periodic travelling waves of the modified KdV equation and rogue waves on the periodic background", J. Nonlin. Sci. {\bf 29} (2019) 2797--2843.

\bibitem{ChenPelin-BO} J. Chen and D.E. Pelinovsky, ``Bright and dark breathers of the Benjamin-Ono equation on the traveling periodic background", Wave Motion  (2024) in press.

\bibitem{CPW} J. Chen, D.E. Pelinovsky, and R.E. White, ``Periodic standing waves in the focusing nonlinear Schr\"{o}dinger equation: Rogue waves and modulation instability", Physica D {\bf 405} (2020) 132378 (13 pages).

\bibitem{Feng} B.F. Feng, L. Ling, and D.A. Takahashi,
``Multi-breathers and high order rogue
waves for the nonlinear Schr\"{o}dinger equation on the elliptic function
background", Stud. Appl. Math. {\bf 144} (2020) 46--101.

\bibitem{Gavr} S. Gavrilyuk and K. M. Shyue, ``Singular solutions of the BBM equation: analytical and numerical study", Nonlinearity {\bf 35} (2022)  388--410


\bibitem{Gesztesy} 	F. Gesztesy and R. Svirsky, ``(m)KdV solitons on the background of quasi-periodic finite-gap solutions", Mem. Amer. Math. Soc. 
{\bf 118} (1995), 563 (88 pages)

\bibitem{Grimshaw} R. Grimshaw, ``Nonlinear wave equations for the oceanic internal solitary waves", Stud. Appl. Math. {\bf 136} (2016), 214--237.

\bibitem{HMP} M. Hoefer, A. Mucalica, and D. E. Pelinovsky, ``KdV breathers on cnoidal wave background", J. Phys. A: Mathem. Theor. {\bf 56} (2023) 185701 (25 pages)

\bibitem{Kuznetsov} E. A. Kuznetsov and A. V. Mikhailov, ``Stability of stationary waves in nonlinear weakly dispersive
media", Sov. Phys. JETP {\bf 40} (1974)  855--859

\bibitem{Lawden} D. F. Lawden, {\em Elliptic Functions and Applications}, (Springer New York, NY, 2013)

\bibitem{Geng} R. Li and X. Geng, ``Rogue waves and breathers of the derivative Yajima--Oikawa long wave-short wave equations on theta-function backgrounds", 
J. Math. Anal. Appl. {\bf 527} (2023) 127399 (26 pages)

\bibitem{Ling} L. Ling and X. Sun, ``Multi-elliptic-dark soliton solutions in the defocusing nonlinear Schr\"{o}dinger equation", Appl. Math. Lett. {\bf 148} (2023) 108866 (9 pages)

\bibitem{Sun1} L. Ling and X. Sun, ``Stability of elliptic function solutions for the focusing modified KdV equation", {\it Advances in Mathematics}, {\bf 435} (2023) 109356 

\bibitem{Sun2} L. Ling and X. Sun, ``The multi elliptic-breather solutions and their asymptotic analysis for the MKDV equation",
Stud. Appl. Math. {\bf 150} (2023) 135--183

\bibitem{Lou} S. Lou, X. Cheng, and X. Tang, ``Dressed dark solitons 
of the defocusing nonlinear Schr\"{o}dinger equation", Chinese Phys. Lett. {\bf 31} (2014) 070201 (4 pages)

\bibitem{Maiden}  M. D. Maiden, D. V. Anderson, A. A. Franco,
G. A. El, and M. A. Hoefer, ``Solitonic dispersive hydrodynamics:
Theory and observation,'' Phys. Rev. Lett. {\bf 120} (2018) 144101 

\bibitem{Mao} Y. Mao, S. Chandramouli, W. Xu, and M. Hoefer, ``Observation of traveling breathers and their scattering in a two-fluid system", Phys. Rev. Lett. {\bf 131} (2023) 147201 (7 pages)

\bibitem{Matveev} V.B. Matveev and M. A. Salle, {\em Darboux Transformations and Solitons} (Springer-Verlag, Berlin, 1991).

\bibitem{Pelinovsky}  E. Pelinovsky, T. Talipova, I. Didenkulova, and E. Didenkulova, ``Interfacial long traveling waves in a two-layer fluid with variable depth", Stud. Appl. Math. {\bf 142} (2019), 513--527.

\bibitem{Ralph} E.A. Ralph, L. Pratt, "Predicting eddy detachment for an equivalent barotropic thin jet", J Nonlinear Sci. {\bf 4} (1994), 355–374. 

\bibitem{Shin} H. Shin, ``The dark soliton on a cnoidal wave background", 
J. Phys. A: Math. Gen. {\bf 38} (2005) 3307--3315

\bibitem{Smirnov} O. Smirnov, ``The Dirac operator with elliptic potential", 
Sb. Math. {\bf 186} (1995) 1213--1221

\bibitem{Sprenger} P. Sprenger, M. A. Hoefer, and G. A. El, ``Hydrodynamic optical soliton tunneling", Phys. Rev. E {\bf 97} (2018)  032218

\bibitem{Takahashi} D. A. Takahashi, ``Integrable model for density-modulated quantum condensates: Solitons passing through a soliton lattice``, Phys. Rev E. {\bf 93} (2016) 062224 (20 pages).

\bibitem{Vainshtein} A. Vainchtein, ``Solitary waves in FPU-type lattices", 
Physica D {\bf 434} (2022) 133252 (22 pages)

\bibitem{Sande} K. van der Sande, G. A. El and M. A. Hoefer, ``Dynamic soliton–mean flow interaction with non-convex flux", J. Fluid Mech. {\bf 928} (2021) A21 (43 pages)

\bibitem{ZS} V.E. Zakharov and A.B. Shabat, “Exact theory of two-dimensional self-focusing and one-dimensional selfmodulation of waves in nonlinear media”, Sov. Physics JETP {\bf 34} (1972), 62--69

\end{thebibliography}
\end{document}